\documentclass[12pt]{amsart}
\usepackage{a4}
\usepackage{graphicx}
\usepackage{amsmath,amssymb,amsthm, color, fancyvrb, array, enumerate}
\usepackage{amsaddr}

\theoremstyle{plain}
\newtheorem{thm}{Theorem}[section]
\newtheorem{lem}[thm]{Lemma}
\newtheorem{cor}[thm]{Corollary}
\newtheorem{prop}[thm]{Proposition}
\theoremstyle{definition}
\newtheorem{defn}[thm]{Definition}

\newtheorem{assumpt}[thm]{Assumption}

\newcommand{\R}{{\mathbb R}}

\newcommand{\tr}{\mathrm{tr}}
\newcommand{\di}{\mathrm{div}}
\newcommand{\p}{\mathcal{P}}
\newcommand{\D}{\mathbb{D}}
\newcommand{\N}{\mathbb{N}}
\newcommand{\WMN}{{\partial M_\mathbb{N}}}
\newcommand{\WMD}{{\partial M_\mathbb{D}}}
\newcommand{\MN}{{\partial M_{N}}}
\newcommand{\MD}{{\partial M_{D}}}
\newcommand{\vol}{\textrm{vol}}

\title[Constraint Equations on Compact with Boundary]{The Einstein Constraint Equations on Compact Manifolds with Boundary}
\author{{James Dilts}}
\address{Department of Mathematics, University of Oregon\\ Eugene, OR 97403}
\email{jdilts@uoregon.edu}
\date{\today}

\begin{document}
\maketitle
\begin{abstract} We continue the study of the Einstein constraint equations on compact manifolds with boundary initiated by Holst and Tsogtgerel. In particular, we consider the full system and prove existence of solutions in both the near-CMC and far-from-CMC (for Yamabe positive metrics) cases. We also make partial progress in proving the results of previous ``limit equation" papers by Dahl, Gicquaud, Humbert and Sakovich.
\end{abstract}

\section{Introduction}

General relativity can be considered the study of Lorentzian manifolds $(\widetilde M^n, h)$, called spacetimes, that satisfy the Einstein equations. In the vacuum case, these equations reduce to $\textrm{Ric}_h = 0$, i.e., that $(\widetilde M,h)$ is Ricci flat. One way to approach this study is by considering which $n-1$ Riemannian manifolds $(M,\widetilde g)$ can be isometrically embedded in some spacetime $(\widetilde M, h)$ with a given second fundamental form $K$. A necessary condition for this to occur is that $\widetilde g$ and $K$ satisfy the Einstein constraint equations,
\begin{align}\label{eq:constraints1}
R_{\widetilde g} &= |K|^2_{\widetilde g} - (\tr_{\widetilde g} K)^2\\ \label{eq:constraints2}
0 &= \di_{\widetilde g} K - \nabla \tr_{\widetilde g} K,
\end{align} where $R_{\widetilde g}$ is the scalar curvature of ${\widetilde g}$. Choquet-Bruhat showed in \cite{CB52} that this condition is in fact also sufficient to produce a spacetime into which $(M, \widetilde g)$ embeds.

The main method used in trying to understand the full set of triples $(M, \widetilde g, K)$ that can be thus embedded is called the conformal method. It was developed by Lichnerowicz, Choquet-Bruhat and York. Let $N=\frac{2n}{n-2}$ and let $L$ be the conformal Killing operator,
\[
LW_{ij} = \nabla_i W_j + \nabla_j W_i - \frac{2}{n} \nabla^k W_k g_{ij}.
\] The conformal method then decomposes $\widetilde g$ and $K$ as
\begin{equation}\label{decomposition}
\begin{array}{c}
\widetilde g = \phi^{N-2} g \\
K = \frac\tau{n} \phi^{N-2} g_{ij} + \phi^{-2}(\sigma_{ij} + LW_{ij}),
\end{array}
\end{equation} where $g$ is a Riemannian metric, $\sigma$ is a trace-free, divergence-free symmetric 2-tensor, $\tau$ and $\phi$ are scalar functions and $W$ is a vector field. Note that $\tau$ can be interpreted as the mean curvature of $M$ in the spacetime $\widetilde M$.

Using the decomposition \eqref{decomposition} in the constraint equations \eqref{eq:constraints1}-\eqref{eq:constraints2}, a calculation reduces the constraint equations to
\begin{equation}\label{origLich}
-\frac{4(n-1)}{n-2} \Delta \varphi + R_{g} \varphi +\frac{n-1}{n}\tau^2 \varphi^{N-1} - |\sigma +LW|^2 \varphi^{-N-1} = 0,
\end{equation}
\begin{equation}\label{origVect}
\di LW - \frac{n-1}{n} \varphi^{N} d\tau =0,
\end{equation} where $\Delta$ is the Laplacian with negative eigenvalues. Equation \eqref{origLich} is known as the Lichnerowicz equation, equation \eqref{origVect} is often known as the vector equation, while together they are known as the conformal constraint equations.

With this reformulation, we can study the set of triples $(M, \widetilde g, K)$ by asking which ``seed data" $(M, g, \sigma, \tau)$ lead to solutions $(\phi, W)$ of the conformal constraint equations \eqref{origLich}-\eqref{origVect}.

The simplest case is when $d\tau \equiv 0$, which can be interpreted as $M$ having constant mean curvature (CMC). In this case the conformal constraint equations decouple, making the system much easier to solve. Following work from York, O'Murchadha, Choquet-Bruhat and others, Isenberg completed the classification of seed data on closed manifolds leading to solutions in \cite{Isenberg95}.

Since then much progress has been made, both in considering other types of manifolds and in loosening the restriction on the mean curvature. Hyperbolic \cite{IP97,GS12}, asymptotically Euclidean \cite{CBIY00, DGI13}, asymptotically cylindrical \cite{CM12,CMP12,Leach13} and compact with boundary \cite{HT13} manifolds have now been considered. The case when the mean curvature is near constant (i.e., the near-CMC condition) is well understood for closed manifolds (see \cite{IM96,IOM04,ACI08}), and progress has been made in other cases as well (such as in this paper or \cite{IP97,GS12,DGI13,Leach13}). The far-from-CMC case resists analysis, but limited results have been achieved, originally by Holst, Nagy and Tsogtgerel in \cite{HNT08} and extended by Maxwell in \cite{Maxwell09}. However, these results unfortunately instead require $|\sigma|$ to be sufficiently small. It is currently unknown whether both $|\sigma|$ and $d\tau$ can be large. For a nice review of the constraints, though leaving out the most recent progress, see \cite{BI04}.

In this paper, we consider compact manifolds with boundary. Physically, these can be seen as pieces of larger spacelike slices of a spacetime, since we don't have any reason to suspect the universe has a boundary. Also, compact manifolds with boundary appear naturally in some numerical computations where an exterior boundary condition can be an approximation of an asymptotically flat end, while interior boundary conditions arise from black hole excision, cf. \cite{Gourgoulhon12}.

We extend the results of several different papers to this new situation. In doing this, we are indebted to the groundwork laid by Holst and Tsogtgerel in \cite{HT13}, where they considered the Lichnerowicz equation \eqref{origLich} alone. We will extend the results of Holst, Nagy, Tsogtgerel and Maxwell in finding far-from-CMC solutions (cf. \cite{HNT08,Maxwell09}), as well as extend their methods, such as proving a global sub/supersolution existence theorem and using a Green's function to show that only a supersolution is actually needed in many cases. We also make partial progress in proving the results of previous ``limit equation'' papers such as \cite{DGH11, GS12,DGI13}. We were unable, however, to complete the final step in deriving the limit equation.

Holst, Meier and Tsogtgerel have recently written a paper \cite{HMT13} that is similar to this paper in several ways. It also finds solutions to the constraint equations on compact manifolds with boundary, and with lower regularity than in this paper. They also handle the coupling of the boundary conditions in a slightly different manner. However, they do not include the Green's function results nor the limit equation results.

\section{Setup}

In this section, we set up a general boundary value formulation of the conformal constraint equations. We introduce a number of pieces of notation, list our standard assumptions, and then list the associated boundary value problems for the conformal constraint equations in \eqref{lichSystem}-\eqref{vectorSystem}.

The boundary conditions for solving the Einstein constraint equations on compact manifolds with boundary can be fairly complicated. For instance, the conditions near a black hole in order to have a trapped surface are best represented by a Robin boundary condition. If we are taking a compact piece of an asymptotic manifold, a Dirichlet condition might be better. To allow for greater generality, we split the boundary of the manifold into two pieces in two different ways.

Let $\partial M = \MD \cup \MN$, $\MD \cap \MN = \emptyset$. The scalar field $\phi$ will hold a Dirichlet condition on $\MD$ and a Robin condition on $\MN$. Similarly, let $\partial M = \WMD \cup \WMN$, $\WMD \cap \WMN = \emptyset$. The vector field $W$ will hold a Dirichlet condition on $\WMD$ and a Neumann condition on $\WMN$. Though, in general, we would expect $\WMD = \MD$ and $\WMN = \MN$, we do not require this.

We will be working with functions in Sobolev spaces like $W^{s,p}$. Since these functions are only defined up to a set of measure zero, they do not normally have well defined boundary values. Let $\gamma$ be the trace operator, which gives boundary values for functions in $W^{s,p}$. We will let $\gamma_N$, for instance, be the trace on $\MN$. These maps, $\gamma_N$, $\gamma_D$, $\gamma_\N$ and $\gamma_\D$, are continuous and surjective maps $W^{s,p} \to W^{s-\frac{1}{p},p}(\partial M_i)$ for the appropriate subscript. (Sobolev spaces without specified domains mean over $M$.) Let $\nu$ be the unit (outward) normal on all of $\partial M$. Precomposing the boundary maps with $\partial_\nu$ or other similar derivative operators also gives continuous surjective maps, but to $W^{s-1-\frac{1}{p},p}(\partial M_i)$, as long as $s-1/p$ is not a integer (which can be avoided by reducing $p$ slightly).

We can then formulate the conformal constraint equations in a general way as follows. Let
\[
a_R = \frac{n-2}{4(n-1)} R, \hspace{5mm} a_\tau = \frac{n(n-2)}{4}\tau^2, \hspace{5mm} a_w = \frac{n-2}{4(n-1)}|\sigma+LW|^2.
\] Let $\phi_D>0$ be a function on $\MD$. Let $b_H, b_\theta, b_\tau, b_w$ be functions on $\MN$. We introduce the nonlinear operator
\[
f = \widetilde f \circ \gamma_N
\] where $\widetilde f$ is defined by
\[
\widetilde f(\phi) = b_H \phi + b_\theta \phi^e + b_\tau \phi^{N/2} + b_w\phi^{-N/2}
\] where $e\in\R$. This complicated form for $f$ is to allow a wide range of interesting boundary conditions, and was formulated by Holst and Tsogtgerel in \cite{HT13}. There are interesting boundary conditions that allow the $b_i$ to depend on $LW(\nu,\nu)$, which in turn depends on $\phi$, but this can make solving the combined system much harder since it couples the boundary conditions. Thus, we will assume that the $b_i$ depend only on the seed data $(M, g, \sigma, \tau)$, i.e., the $b_i$ are independent of $\phi$ directly or indirectly.

Let $BW := LW(\nu,\cdot)$ for vector fields $W$. Let $X$ be a vector field on $M$ of the form
\[
X= \sum_i c_i \phi^{k_i},
\] a finite sum, for some vectors $c_i$ of the seed data and some nonnegative real numbers $k_i$. For example, the standard choice is $X= \frac{n-1}{n} \phi^N d\tau$. Let $X_\D$ be a vector field on $\WMD$ and $X_\N$ be a one-form on $\WMN$. For simplicity, we will require that neither $X_\D$ nor $X_\N$ depend on $\phi$, and that $\sigma(\nu, X_\D) = 0$.

Except for Section \ref{sec:LimitEquation} where we need a little more regularity, we will assume the same regularity throughout the paper. Namely, we make the following set of assumptions.
\begin{assumpt} \label{assumptions}(Standard assumptions)

\begin{itemize}
\item $(M^n,g)$ is a smooth compact manifold with boundary with metric $g \in W^{2,p}$ where $p>n\geq 3$. Thus $R\in L^p$.

\item $g$ has no nontrivial conformal Killing fields $W\in W^{2,p}$ with $BW = 0$ on $\WMN$ and $W = 0$ on $\WMD$.

\item $\tau\in W^{1,p}$ which gives that $\tau\in L^\infty$ and $a_\tau \in L^p$.

\item $|\sigma|^2 \in W^{1,p}$, which gives that $a_w \in L^p$ if $|LW|^2 \in L^p$ as well. Also, $\sigma\cdot \nu = 0$ on $\WMN$, where $\nu$ is the normal vector to the boundary. Such tensors are shown to exist in subsection \ref{sec:YorkDecomposition}.

\item $b_H, b_\theta, b_\tau, b_w \in W^{1-\frac{1}{p},p}(\MN)$ and only depend on the seed data.

\item $(e-1)b_\theta \geq 0$ with $e\neq 1$, $b_\tau\geq 0$ and $b_w \leq 0$.

\item $b_\theta$ and $b_H - \frac{n-2}{2}H$ have constant sign on each component of the boundary, where $H$ is the mean curvature of the boundary.

\item $\phi_D \in W^{2-\frac{1}{p},p}(\MD)$ with $\phi_D>0$.

\item The Lichnerowicz problem, as defined below, is conformally covariant, as defined in Section \ref{sec:LichnerowiczProblem}.

\item The coefficients $c_i$ from $X$ are in $L^p$.

\item $X_\N\in W^{1-\frac{1}{p},p}(\WMN)$, $X_\D\in W^{2-\frac{1}{p},p}(\WMD)$ and $\sigma(\nu,X_\D) = 0$.
\end{itemize}
\end{assumpt}
Most of the results in this paper do not require all of the assumptions listed above. For instance, the results from we cite from \cite{HT13} hold for $g\in W^{s,p}$ with $s\geq 1$ and $s>n/p$. However, our main results require all of them, and so for simplicity we assume these assumptions throughout this paper.

Let $[\phi_-,\phi_+]_{2,p}:= \{\phi\in W^{2,p}: 0<\phi_-\leq \phi \leq \phi_+ \text{ a.e.}\}$. The standard regularity conditions give that $f$ is a map from $[\phi_-,\phi_+]_{2,p} \to W^{1-\frac{1}{p},p}(\MN)$.

We then split the conformal constraint equations \eqref{origLich}-\eqref{origVect} into two problems and consider them, at first, separately. The Lichnerowicz problem is to find an element $\phi\in [\phi_-,\phi_+]_{2,p}$ such that
\begin{equation}\label{lichSystem}
F(\phi) := \begin{array}{rlc} -\Delta\phi + a_R\phi + a_\tau \phi^{N-1} - a_w \phi^{-N-1} &\hspace{-2.4mm}= 0\\
\gamma_N\partial_\nu\phi + f(\phi) &\hspace{-2.4mm}=0 &  \textrm{ on } \MN\\
\gamma_D\phi-\phi_D & \hspace{-2.4mm}= 0  &  \textrm{ on } \MD.\end{array}
\end{equation}

The vector problem is then to find an element $W\in W^{2,p}$ such that
\begin{equation}\label{vectorSystem}
\p^{s,p}(W) := \begin{array}{rlc}\di L W &\hspace{-2.4mm}= X & \\
\gamma_\N BW &\hspace{-2.4mm}= X_\N &  \textrm{ on } \WMN \\
\gamma_\D W &\hspace{-2.4mm}= X_\D & \textrm{ on } \WMD. \end{array}
\end{equation}

As in the introduction, if we can simultaneously solve these two problems, we can reconstruct a solution to the Einstein constraint equations \eqref{origLich}-\eqref{origVect}.

Before we discuss previous results about this system, we need a Yamabe classification theorem. Escobar, in \cite{Escobar92}, showed that in many cases, one could conformally transform a metric on a compact manifold with boundary to one with constant scalar curvature and minimal (mean curvature zero) boundary. Brendle and Chen, in \cite{BC09}, expanded the list of allowable manifolds. This general problem remains unsolved. Fortunately, Holst and Tsogtgerel proved a weaker version of this classification that suffices for our needs.

\begin{thm}\cite[Thm 2.2]{HT13}\label{YamabeClassification}
The metric $g$ is in exactly one of $Y^+,Y^0,Y^-$, where $g\in Y^+$ ($\in Y^0, \in Y^-$) means that there is a metric in the conformal class of $g$ whose scalar curvature is continuous and positive (resp. zero or negative), and boundary mean curvature is continuous and has any given sign (resp. is identically zero, has any given sign). ``Any given sign" includes the case that it is identically zero.
\end{thm}

Following the closed case, we say $g$ is in the positive Yamabe class if $g\in Y^+$, and similarly for the other classes.

\section{Lichnerowicz Problem}\label{sec:LichnerowiczProblem}

Now we can give some results about the Lichnerowicz problem (\ref{lichSystem}) from \cite{HT13}. First, one of the most successful methods of finding solutions to the Lichnerowicz equation has been the method of sub and supersolutions. The appropriate generalization for this problem is as follows.
\begin{thm}\cite[Thm 5.1]{HT13}\label{lichExistence}
Let $\phi_-,\phi_+\in W^{2,p}$ be such that $F(\phi_+)\geq 0$ and $F(\phi_-) \leq 0$ (i.e., are super and subsolutions respectively, see \eqref{lichSystem}), and such that $0<\phi_-\leq \phi_+$. Then there exists a positive solution $\phi \in [\phi_-,\phi_+]_{2,p}$ of the Lichnerowicz problem \eqref{lichSystem}.
\end{thm}

One nice property of the Lichnerowicz equation, i.e., the first line of \eqref{lichSystem}, is that it is conformally covariant. For example, if we have a supersolution, we can do a conformal transformation in a particular way, and the supersolution multiplied by the conformal factor will still be a supersolution. Similarly, the Dirichlet part of the boundary condition will also be conformally covariant. However, the Neumann/Robin part of the boundary condition will not always be.

Let $\psi$ be a conformal factor, and let hats denote transformed quantities. In particular, we set $\widehat g = \psi^{N-2} g$, with scalar curvature $\widehat R$ and boundary mean curvature $\widehat H$. Let $\widehat \tau = \tau$, $\widehat\sigma_{ij} = \psi^{-2} \sigma_{ij}$, $\widehat{LW}_{ij} = \psi^{-2}LW_{ij}$ and $\widehat\phi_D = \phi_D/\psi$. Since we want the $b_i$ data to be general, simply let there be some transformation rule for them like for the other data. Let $\widehat F$ be the operator $F$ from \eqref{lichSystem}, but with the hatted data.
\begin{defn}\label{conformalCovariance} We say the Lichnerowicz problem \eqref{lichSystem} is conformally covariant if
\begin{align*}
\widehat F(\phi) = 0 & \,\,\,\,\Leftrightarrow  \,\,\,F(\psi\phi) = 0 \\
\widehat F(\phi) \geq 0 & \,\,\,\,\Leftrightarrow  \,\,\,F(\psi\phi) \geq 0 \\
\widehat F(\phi) \leq 0 & \,\,\,\,\Leftrightarrow  \,\,\,F(\psi\phi) \leq 0 \\
\end{align*} for any positive conformal factor $\psi$.
\end{defn}

We will next consider when the Lichnerowicz problem is conformally covariant. Recall that
\[
\widehat R = \psi^{2-N} R - \frac{4(n-1)}{n-2} \psi^{1-N}\Delta\psi,
\]\[
\widehat \Delta \phi = \psi^{2-N} \Delta\phi + 2 \psi^{1-N}\langle d\psi, d\phi\rangle_g.
\] Using these two, we can show that the Lichnerowicz equation, the first line of \eqref{lichSystem}, is conformally covariant. In particular, if we let $F_1$ be the first line of \eqref{lichSystem}, and $\widehat F_1$ be the same operator with the transformed quantities, we have $\widehat F_1(\phi) = \psi^{1-N} F_1(\psi\phi)$.

Let $F_3$ be the Dirichlet boundary condition operator, i.e., the third line of \eqref{lichSystem}. Since the right side is a fixed function, it is particularly easy to show conformal covariance. It is clear that $\widehat F_3(\phi) = \psi^{-1}F_3(\psi\phi)$.

The Neumann boundary condition, the second line of \eqref{lichSystem} is the most difficult, mostly because it was purposely designed to be general. In different cases the coefficients might be very different, and transform differently. Recall that
\[
\widehat H = \psi^{1-N/2} H + \frac{2}{n-2}\psi^{-N/2} \partial_\nu \psi
\] for boundary mean curvature $H$, and
\[
\partial_{\widehat\nu} \phi = \psi^{1-N/2} \partial_{\nu}\phi.
\] Together, these show that
\begin{equation}\label{eq:boundaryConformalCovariance}
\partial_{\widehat\nu} \phi + \frac{n-2}{2} \widehat H \psi = \psi^{-N/2}\left(\partial_\nu (\psi\phi) + \frac{n-2}{2} H (\psi\phi)\right).
\end{equation} Thus, if $b_H = \frac{n-2}{2} H$ we have a good start towards conformal covariance.

Holst and Tsogtgerel in \cite{HT13} list a number of possibly useful boundary conditions. We will consider each one in turn. We will not present the details of each, but describe it briefly and consider if it is conformally covariant.

The first condition represents a Robin condition for compact sections of an asymptotically Euclidean manifold. For this condition, $b_H = (n-2)H$, $b_\theta = -(n-2)H$ with $e=0$, and $b_\tau = b_w = 0$. If we attempt a conformal transformation we get
\[
\partial_{\widehat\nu} \phi + (n-2) \widehat H \phi - (n-2) \widehat H = \psi^{1-N/2} \partial_\nu \phi + (n-2)\psi^{1-N/2} (\phi-1)H + 2(\phi-1)(\psi^{-N/2}\partial_\nu \psi).
\] Thus this boundary condition is not conformally covariant.

Another possibility is a boundary condition that makes the boundary a minimal surface. Here, $b_\theta = b_\tau = b_w = 0$ and $b_H = \frac{n-2}{2}H$. This is exactly the case we've already considered, and so this condition is conformally covariant.

The next condition guarantees the existence of trapped surfaces. Let $b_H = \frac{n-2}{2}H$, $b_\theta = \pm \frac{n-2}{2(n-1)}\theta_{\pm}$, $b_\tau = \mp \frac{n-2}{2}\tau$, and $b_w = \pm \frac{n-2}{2(n-1)} S(\nu, \nu)$, where $\theta_\pm$ are the expansion scalars and $S = \sigma + LW$. Comparing exponents, we see that $\theta_\pm$ must transform as $\widehat\theta_\pm = \theta_\pm \psi^{e-N/2}$. Fortunately, this is exactly the transformation described in \cite{HT13}. Similarly, $\tau$ must transform as $\widehat \tau = \tau$, which is fortunately the same as was required for the main Lichnerowicz equation. Finally, since $S = \sigma+LW$, $\widehat S(\widehat \nu, \widehat \nu ) = \psi^{-N} S(\nu,\nu)$, which is the proper transformation. Thus this boundary condition is conformally covariant.

Lastly, we have a different formulation that also guarantees the existence of trapped surfaces. In particular, $b_H = \frac{n-2}{2} H$ and $b_\theta = (\theta_+ - \theta_-)$ with $e$ arbitrary and $b_w = b_\tau = 0$. This is conformally covariant for the same reasons as the previous boundary condition.

In general, $b_H = \frac{n-2}{2}H$ is required for $F_2$, the second line of \eqref{lichSystem}, to be conformally covariant. This is required so that the $\partial_{\widehat \nu}\phi$ can transform correctly, as in \eqref{eq:boundaryConformalCovariance}. Also, any nonzero quantities of $b_\theta, b_\tau$ and $b_w$ must transform such that something like $\widehat b_\theta \phi^e = \psi^{-N/2} b_\theta (\psi\phi)^e$ holds for functions $\phi$. Thus, the Lichnerowicz problem \eqref{lichSystem} is conformally covariant if either $\MN = \emptyset$ or there is a restriction on the coefficients $b_i$ of the Neumann boundary condition.

The easiest case to solve the Lichnerowicz problem \eqref{lichSystem} is in the so-called ``defocusing case," which restricts the signs of most of the coefficients. In  particular, the defocusing case means that $a_\tau\geq 0$, $a_w\geq 0$, $(e-1)b_\theta \geq 0$ with $e\neq 1$, $b_\tau\geq 0$ and $b_w \leq 0$. While the first two requirements are natural, the other restrictions are made primarily for ease of solving. However, they do include most of the important boundary conditions, including the ones we will care about. In the defocusing case, Holst and Tsogtgerel proved fairly exhaustive existence results for the Lichnerowicz problem in \cite[Thm 6.1,2]{HT13}, as well as uniqueness in \cite[Lem 4.2]{HT13}. They also proved the continuity of the solution map for the Lichnerowicz equation. We include the theorem here for use in Section \ref{sec:LimitEquation}.

\begin{lem}\label{LichnerowiczContinuity}\cite[Thm 8.1]{HT13}
Let $\alpha = (a_\tau, a_w, b_H, b_\tau, b_\theta, b_w, \phi_D)$, with regularity
\[
\alpha \in \left[L^p\right]^2 \times \left[W^{1-\frac{1}{p},p}(\MN)\right]^4 \times W^{2-\frac{1}{p},p}(\MD).
\] Assume moreover that the solution map of the solution map of Lichnerowicz problem \eqref{lichSystem} (as a map $\alpha \mapsto \phi \in W^{2,p}$)  is well defined at $\alpha$. Then the solution map is defined in a neighborhood of $\alpha$ and is (Fr\'echet) differentiable there provided that either $\MD \neq \emptyset$ or at least one of $a_\tau +a_w$, $b_\tau$, $b_\theta$ or $b_w$ is not identically zero.
\end{lem}

\section{Vector Problem}

For the vector problem \eqref{vectorSystem}, the following estimate holds.

\begin{thm} \label{schauderEstimateTheorem}
If $W\in W^{2,p}$ and satisfies the system \eqref{vectorSystem}, then $W$ satisfies the estimate
\begin{equation} \label{schauderEstimate}
\|W\|_{2,p} \leq C\left(\|X\|_{p} + \|X_\N \|_{W^{1-\frac{1}{p},p}(\WMN)} + \|X_\D\|_{W^{2-\frac{1}{p},p}(\WMD)}\right)
\end{equation} where $\|\cdot\|_{2,p}$ and $\|\cdot\|_p$ are the $W^{2,p}$ and $L^p$ norms respectively. Recall that $X = \sum c_i \phi^{k_i}$, and assume that $\phi$ has an upper bound $B\geq 1$. Let $k = \sup k_i$. Then $LW$ satisfies the estimate
\begin{equation}\label{schauderEstimateLW}
\|LW\|_{1,p} \leq C(B^k+1)
\end{equation} for some $C$ independent of $W$ and $\phi$.
\end{thm}
\begin{proof}
Estimate \eqref{schauderEstimate} is \cite[Prop 4]{Maxwell05b}. Estimate \eqref{schauderEstimateLW} uses \eqref{schauderEstimate}, $\|LW\|_{1,p} \leq \|W\|_{2,p}$ and that $X_\N$ and $X_\D$ are independent of $\phi$. The $C$ in this inequality depends on the embedding constants, the $c_i$, $X_\N$ and $X_\D$. Note that we use that none of the $k_i$ are negative.
\end{proof}

Let \[
\p^{2,p}: W^{2,p} \to L^{p}\times W^{1-\frac{1}{p},p}(\WMN) \times W^{2-\frac{1}{p},p}(\WMD)\] be the map $W \mapsto (\di L W, \gamma_\N BW, \gamma_\D W)$, as in \eqref{vectorSystem}. A standard result (see \cite[Lem B.5]{HT13}) gives that estimate \eqref{schauderEstimate} implies that $\p^{2,p}$ is semi-Fredholm under our assumptions on $p$.

If $W$ is a vector field on $M$ such that $LW = 0$ on $M$, $BW = 0$ on $\WMN$ and $W = 0$ on $\WMD$, we say that $W$ is a conformal Killing field with zero boundary condition.

\begin{thm}\label{vectExistence}
The operator $\p^{2,p}$ is an isomorphism.
\end{thm}
\begin{proof}
We proceed as in \cite{Maxwell05b}. We first assume that $g$ is any smooth metric; the desired results will then follow from an index theory argument.

We only need to prove that $\p^{2,2}$ is invertible. Indeed, if something is in the kernel of $\p^{2,2}$, we know by elliptic regularity that it is in $W^{2,p}$, and so must be in the kernel of $\p^{2,p}$ also. Also, if $\p^{2,2}$ is surjective, then its image certainly contains $C^\infty_c\times C^\infty(\WMN)\times C^\infty(\WMD)$. Using elliptic regularity again, the image of $\p^{2,p}$ will also contain that space. Since the image of $\p^{2,p}$ is closed (since it is semi-Fredholm), we also have that $\p^{2,p}$ is surjective by the density of $C^\infty$ in Sobolev spaces.

So we now restrict our attention to $\p = \p^{2,2}$. To show $\p$ is injective, we show that any element of the kernel must be a conformal Killing field. Suppose $u \in \ker \p$. We then integrate by parts and find
\[
0 = -\int_M \langle \di Lu ,u \rangle = \int_M \langle Lu, Lu \rangle + \int_{\partial M} Lu(\nu, u)
\] where $\nu$ is the unit normal to $M$. Since $u$ is in the kernel, either $u$ or $Bu$ is 0 on each component of the boundary and so we get that $Lu \equiv 0$. Thus $u$ is a conformal Killing field with zero boundary condition, and is smooth by elliptic regularity. By Assumption \ref{assumptions} any smooth conformal Killing field with zero boundary condition must be trivial. Thus $\p$ is injective.

To show $\p$ is surjective, we can instead show that the adjoint $P^*$ is injective by \cite[19.2.1]{Hormander}. The dual space of $L^2\times H^{1/2}(\WMN) \times H^{3/2}(\WMD)$ is $L^2\times H^{-1/2}(\WMN) \times H^{-3/2}(\WMD)$. From elliptic regularity and rescaled interior estimates, we know that if $\p^*(f_1, f_2, f_3) = 0$, then in fact the $f_i$ are smooth (cf. \cite[19.2.1]{Hormander}). For smooth $\theta$, integrating by parts gives
\begin{align*}
0 &= \langle \p^*(f_i), \theta\rangle \\
&= \int_M \langle \di Lf_1, \theta\rangle + \int_{\partial M} \left(L\theta(\nu,f_1) - Lf_1(\nu,\theta)\right) + \int_{\WMN} L\theta(\nu,f_2) + \int_{\WMD} f_3 \theta.
\end{align*} By using $\theta$ that are zero on the boundary, we can immediately see that $\di L f_1 = 0$ in $M$. As shown in Lemma \ref{ODE} below, one can readily show that if $\omega$ is a smooth 1-form on $\partial M$ and $\psi$ is a smooth function on $\partial M$ that there exists a $\theta\in C^\infty$ such that $\theta = \psi$ and $B\theta = \omega$ on $\partial M$. Thus it immediately follows that $Bf_1 = 0$, $f_1=-f_2$ on $\WMN$ and $f_1 = 0$, $Bf_1 = f_3$ on $\WMD$.

Since $\di L f_1 =0$ and either $Bf_1 =0$ or $f_1=0$ on each component of the boundary, by integration by parts $f_1$ must be a conformal Killing field. Similar to earlier, this shows that $f_1 \equiv 0$, and so $f_2$ and $f_3$ must also be zero. Thus $\p^*$ is injective and so $\p$ is an isomorphism.

That was all in the smooth metric case. Suppose $g$ is only in $W^{2,p}$ with $p>n$. To show $\p^{2,p}$ is Fredholm of index 0, it is enough to show its index is 0. Since $g$ can be approximated with smooth metrics $g_k$, and since each $\p^{2,p}_{g_k}$ has index 0, so does the limit $\p^{2,p}$. To show that the kernel of $\p^{2,p}$ consists of conformal Killing fields with zero boundary condition, we integrate by parts again using the fact that $u=0$ or $Bu=0$ on the boundary.

\end{proof}

\begin{lem}\label{ODE}
Let $(M,g)$ be a smooth manifold with boundary $\partial M$, with smooth metric. If $\omega$ is a smooth 1-form on $\partial M$ and $\psi$ is a smooth vector field on $\partial M$ (perhaps including a component in the normal direction), then there exists a vector field $\theta\in C^\infty(M)$ such that $\theta = \psi$ and $B\theta = \omega$ on $\partial M$.
\end{lem}
\begin{proof}
Let $\nu$ be the unit inward normal vector to $\partial M$. Use $\nu$ to pick boundary normal coordinates, where $\nu$ geodesics define one of the coordinates. So, for instance, the boundary has $x^\nu = 0$. To show the desired $\theta$ exists, we will express it as the solution to a local PDE. Taking a solution on a neighborhood of the boundary, and then extending it smoothly, we get the desired $\theta$.

We take the initial conditions $\theta = \psi$ on $\partial M$. Then $L\theta\cdot \nu = \omega$ in local coordinates reduces to
\[
\nabla_\nu \theta_i = f_i
\] for some known terms $f_i$ in terms of $\omega$ and $\nabla_j \theta_k$ for $j\neq \nu$. If we extend $\omega$ by making the coordinate components constant (though we could take any smooth extension), this is a standard PDE with smooth short time existence. This completes the theorem.
\end{proof}

\subsection{York Decomposition}\label{sec:YorkDecomposition}

Now that we have a solution of the vector problem \eqref{vectorSystem}, we can talk about the York decomposition of the second fundamental form. In the closed and asymptotically Euclidean cases, the second fundamental form is decomposed into a trace part, a transverse-traceless (i.e., divergence free and trace free) part and a ``longitudinal"-traceless part. One of the useful properties of this decomposition is that it is orthogonal, i.e., it is a direct sum decomposition. This is because two of the terms are traceless, and because
\[
\int_M \sigma \cdot LW = -\int_M \di \sigma \cdot W = 0
\] since the boundary term disappears and since $\sigma$ is divergence free. However, in the general compact with boundary case, this orthogonality is not automatic.

In particular, when we take that same term and integrate by parts, we get
\[
\int_M \sigma \cdot LW = -\int_M \di \sigma \cdot W  + \int_{\partial M} \sigma(\nu, W),
\] where $\nu$ is the normal vector to the boundary. Thus, if we want the decomposition to be a direct sum decomposition, we need to specify either that $\sigma\cdot\nu = 0$ or that $W = 0$ on each component of $\partial M$.

Thus, we need to construct transverse-traceless symmetric 2-tensors $\sigma$ with that boundary condition. Let $S$ be any symmetric traceless 2-tensor. We assume it is traceless since the removing of the trace is well understood. We then solve the following problem for $V$:
\[
\begin{array}{rlc}\di L V &\hspace{-2.4mm}= \di S & \\
 BV &\hspace{-2.4mm}= S\cdot \nu &  \textrm{ on } \WMN \\
 V &\hspace{-2.4mm}= 0 & \textrm{ on } \WMD \end{array}
\] where traces are implied if the data or solutions are not sufficiently regular.

By Theorem \ref{vectExistence}, we know that there is a $V$ solving this system. We then let $\sigma = S-LV$. Thus we get that $\sigma$ is transverse-traceless, as in the standard York decomposition. In addition we get that $\sigma\cdot \nu = 0$ or $V = 0$ on $\partial M$, and so the decomposition is orthogonal on arbitrary compact manifolds with boundary.

\begin{prop}
There is a direct sum decomposition of traceless symmetric 2-tensors $S$ into transverse-traceless tensors $\sigma$ and ``longitudinal"-traceless tensors $LV$ such that either $\sigma\cdot \nu = 0$ or $V=0$ on each component of the boundary.
\end{prop}

We note that while we do take such a $\sigma$ for our problem, the $W$ we find by solving the vector problem \eqref{vectorSystem} is not the other half of this decomposition. This is because we allowed $X_\D$ to be orthogonal to the fixed $\sigma\cdot \nu$. However, we still have the property that $\int_M \sigma \cdot LW = 0$.

\section{The Combined System}

Next we will show that given a global sub and supersolution (defined below), the combined system admits a solution essentially under the same conditions as the Lichnerowicz problem does alone, as in Theorem \ref{lichExistence}. To do this, we need Theorem 5 from \cite{HNT09}.

\begin{thm}\cite[Thm 5]{HNT09} \label{fixedPoint}
Let $X$ and $Y$ be Banach spaces, and let $Z$ be a real ordered Banach space having the compact embedding $X \hookrightarrow Z$. Let $[\phi_-, \phi_+] \subset Z$ be a nonempty interval which is closed in the topology of $Z$, and set $U = [\phi_-,\phi_+]\cap \bar B_M \subset Z$ where $\bar B_M$ is the closed ball of finite radius $M>0$ in $Z$ around the origin. Assume $U$ is nonempty, and let the maps
\[
S: U \to \mathcal{R}(S) \subset Y, \hspace{5mm} T: U \times \mathcal{R}(S) \to U\cap X,
\] be continuous maps. Then there exist $\phi \in U \cap X$ and $W \in \mathcal{R}(S)$ such that
\[
\phi =T(\phi, W) \,\,\,\, \textrm{ and } \,\,\,\, W = S(\phi).
\]
\end{thm}

Let $W_\phi$ represent the $W^{2,p}$ solution to the vector problem \eqref{vectorSystem} with $\phi$. In general, we expect the functional $F$ (see \eqref{lichSystem}) to depend on $W$, perhaps on both the interior and boundary of $M$. We denote this dependence by $F_W$. We call $\phi_+$ a global supersolution if $F_{W_\phi}(\phi_+) \geq 0$ for any $\phi \in (0,\phi_+]_{2,p}$. We similarly call $\phi_-$ a global subsolution if $F_{W_\phi}(\phi_-) \leq 0$ for any $\phi \in [\phi_-, \phi_+]_{2,p}$.

We call $W$ admissible for a given supersolution $\phi_+$ if $W$ is the solution of the vector problem for some $\phi \in (0,\phi_+]_{2,p}$, and for a super/subsolution set $\phi_+, \phi_-$ if $W$ is the solution of the vector problem for some $\phi \in [\phi_-,\phi_+]_{2,p}$.

\begin{prop} \label{combinedExistence}
Let $\phi_+\in W^{2,p}$ be a global supersolution. Suppose that for any admissible $W$ there exists a subsolution $\phi_- \in W^{2,p}$, with $\phi_-\leq \phi_+$ but not necessarily global. Suppose that any solution $\phi$ of the Lichnerowicz problem \eqref{lichSystem} with an admissible $W$ is bounded below by a uniform constant $K\leq \phi_+$, which may depend on $\phi_+$ and $\phi_-$. Then there exists a positive solution $\phi \in [K,\phi_+]_{2,p}$ and $W \in W^{2,p}$ of the combined conformal constraint system \eqref{lichSystem}-\eqref{vectorSystem}.
\end{prop}
\begin{proof}
\textbf{Step 1.} \emph{Choice of spaces.} We will be using Theorem \ref{fixedPoint}. First, we identify $X = Y = W^{2,p}$ and $Z = W^{\widetilde s,p}$, with $\widetilde s \in \cap (1,2)$ (as in \cite[pg 16]{HT13}). This gives that $X\hookrightarrow Z$ is compact. The ordering on $Z$ is the standard $L^\infty$ ordering, i.e., $f\geq g$ if $f(x)\geq g(x)$ a.e.. Clearly $[K, \phi_+]_{\widetilde s,p}$ is non-empty and closed. Let $U = [K, \phi_+]_{\widetilde s,p} \cap \bar B_M$, with $M$ to be determined in Step 3.

\textbf{Step 2.} \emph{Construction of $S$.} Consider the $X$'s as functions of $\phi$. By Theorem \ref{vectExistence}, $\p^{2,p}$ is an isomorphism.  Let $S = (\p^{2,p})^{-1}\circ (X, X_\N, X_\D): [K, \phi_+]_{\widetilde s,p} \to W^{2,p}$, i.e., the solution map of the vector problem. The continuity of $\phi \mapsto X(\phi)$ is given by Corollary \ref{NormSplitting} since we assumed $X = \sum c_i \phi^{k_i}$. The continuity of $(\p^{2,p})^{-1}$ is given by the estimate \eqref{schauderEstimate}. Thus $S$ is a continuous map.

\textbf{Step 3.} \emph{Construction of $T$.} Let $T(\phi,W)$ be the map $T$ defined in \cite[Thm 5.1]{HT13}, a Picard type map for the Lichnerowicz problem \eqref{lichSystem}. If $W$ is admissible, then $\|a_w\|_{p}$ is bounded by Theorem \ref{schauderEstimateTheorem}. If the scalar and boundary mean curvatures are continuous and of constant sign, the proof of \cite[Thm 5.1]{HT13} then gives the properties for $T$ we need, namely that $T$ is a continuous map in $\phi$ and $W$ and that it maps into $U\cap X$. The choice of $\widetilde s$ is the same as in that proof.

If the curvatures are not continuous and of constant sign, we can use the conformal covariance of the Lichnerowicz problem as in \cite[pg 39]{HNT09} to get the same properties for $T$.

\textbf{Step 4.} \emph{Finish.} We have now fulfilled the hypotheses of Theorem \ref{fixedPoint}, and so there is a solution $\phi \in [K,\phi_+]_{2,p}$ and $W \in W^{2,p}$ to the conformal constraint equations \eqref{lichSystem}-\eqref{vectorSystem}.
\end{proof}

The proof also shows that the same result holds if we don't assume conformal covariance, but instead guarantee that the scalar and mean curvatures are continuous and of constant sign.

\begin{cor}\label{combinedExistenceCor}
Let $\psi$ be a conformal factor independent of $\phi$ and $W$. Suppose the same conditions hold as for Proposition \ref{combinedExistence} except that the global supersolution $\phi_+$, the subsolution(s) $\phi_-$ and uniform lower bound $K$ are for the conformally transformed Lichnerowicz problem $\widehat F$. Then the same existence and regularity holds.
\end{cor}
\begin{proof}
By definition of conformal covariance, if $\phi_+$ is the global supersolution, then $\psi\phi_+$ is a global supersolution of the original Lichnerowicz problem, since $\psi$ does not depend on $W$. Similarly, for any solution $\phi$ of $\widehat F (\phi) = 0$, $F(\psi\phi) = 0$. Thus any solution $\psi \phi$ of the original Lichnerowicz problem is uniformly bounded below by $K\psi$.
\end{proof}

Theorem \ref{combinedExistence} reduces the problem of finding solutions $(\phi, W)$ to the combined conformal constraint equations \eqref{lichSystem}-\eqref{vectorSystem} to that of finding global supersolutions and uniform lower bounds. In fact, in many cases we can reduce the problem to just finding global supersolutions, as in \cite{Maxwell09}. We first prove a lemma.

\begin{lem}\label{boundedSolution}
Let $\alpha \in L^\infty(M)$ with $\alpha\geq 0$ and $\beta \in W^{1-\frac{1}{p},p}(\MN)$ with $\beta\geq 0$. Assume also that either $\alpha \not\equiv 0$, $\beta \not\equiv 0$ or $\MD \neq \emptyset$.  Then there exists constants $c_1$ and $c_2$ such that for every $f\in L^p$, $g\in W^{1-1/p,p}(\MN)$ and $h\in W^{2-1/p,p}(\MD)$, with $f,g,h \geq 0$, the solution $v$ of
\begin{equation}\label{eq:boundSubsolution1}
\begin{array}{rll} -\Delta v + \alpha v &= f& \textrm{   on } M \\
                                \gamma_N \partial_\nu v + \beta v &= g & \textrm{   on } \MN \\
                                \gamma_D v &= h& \textrm{   on } \MD \\ \end{array}
\end{equation} satisfies
\begin{equation}\label{eq:boundSubsolution2}
\sup(v) \leq c_1 \left(\|f\|_{p} + \|g\|_{W^{1-1/p,p}(\MN)} + \|h\|_{W^{2-1/p,p}(\MD)}\right)
\end{equation} and
\begin{equation}\label{eq:boundSubsolution3}
\inf(v) \geq c_2 \left(\int_{M\setminus N} f+ \int_{\MN} g + \int_{\MD} h\right)
\end{equation} where $N$ is any neighborhood of the boundary and $c_2$ depends on $N$.
\end{lem}
\begin{proof}
By our assumptions, the operator acting on $v$ in \eqref{eq:boundSubsolution1} is an isomorphism and thus the first inequality \eqref{eq:boundSubsolution2} holds with the left side replaced by the $W^{2,p}$ norm. By Sobolev embedding, $W^{2,p} \subset L^\infty$ (since $p>n$), and so we get the inequality.

The conditions of Theorem \ref{GreensExistence} are fulfilled, and so let $G(x,y)$ be the Green's function for the operator in \eqref{eq:boundSubsolution1}. Then, since $f,g,h\geq 0$,
\begin{align*}
v(x) &= \int_M f G + \int_\MN g G - \int_\MD h \partial_\nu G\\
&\geq \inf_{M\setminus N}G \int_{M\setminus N} f + \inf_\MN G \int_\MN g +\inf_\MD |\partial_\nu G| \int_\MD h.
\end{align*} The first infimum exists and is nonzero because $G$ is positive away from the boundary. The other infima exist and are nonzero by part (f) of Theorem \ref{GreensExistence}.
\end{proof}

We now proceed to prove our main existence theorems.

\begin{thm}\label{noSubsolutionNeeded}
Let $\phi_+ \in W^{2,p}$ be a global supersolution. Assume that, perhaps after a conformal transformation, $a_R +a_\tau \geq 0$ and $b_H +b_\tau \geq 0$. Assume either that one of those inequalities is strict or that $\MD \neq \emptyset$. Also, assume that either $\sigma \not\equiv 0$, $b_w + b_\theta^- \not \equiv 0$ (where $b_\theta^- = \min\{0, b_\theta\}$) or $\MD \neq \emptyset$. Then there exists $\phi\in W^{2,p}$ and $W\in W^{2,p}$ with $0<\phi\leq \phi_+$ of the combined conformal system \eqref{lichSystem}-\eqref{vectorSystem}.
\end{thm}

Note that since $b_\tau\geq 0$ by assumption (since we are in the defocusing case), the condition on $a_R + a_\tau$ and $b_H + b_\tau$ is easily fulfilled in the case $g\in Y^+$ or the case $g\in Y^0$ and $\tau \not\equiv 0$. However, Theorem \ref{noSubsolutionNeeded} also allows the possibility of $g\in Y^-$ if $g$ has the right curvatures.

\begin{proof}
Let $\psi$ be the conformal factor from the hypotheses that makes $a_{\widehat R} + a_{\widehat \tau}$ and $\widehat b_H + \widehat b_\tau$ nonnegative. As in Section \ref{sec:LichnerowiczProblem}, a hat represents transformed quantities. We transform the quantities as in that section. By Corollary \ref{combinedExistenceCor}, we only need to come up with a subsolution for $\widehat F$, the transformed Lichnerowicz problem, for each admissible $W$, and then show that this family is bounded below uniformly.

Let $v\in W^{2,p}$ be a solution to
\begin{equation}\label{eq:noSubsolution1}
\begin{array}{cl} -\Delta_{\widehat g} v + (a_{\widehat{R}}+ a_{\widehat\tau}) v = a_{\widehat{w}}& \textrm{   on } M \\
         \partial_{\widehat\nu} v + (\widehat b_H + \widehat b_\tau +\widehat b_\theta^+) v = - \widehat b_{w} - \widehat b_\theta^- & \textrm{   on } \MN\\
         v = \widehat\phi_D & \textrm{   on } \MD \\ \end{array}
\end{equation} where the traces are assumed, if necessary, and where $\widehat b_\theta^+ = \max\{0, \widehat b_\theta\}$ and $\widehat b_\theta^- = \min\{0, \widehat b_\theta\}$. We note that since the sign of $b_\theta$ is constant on each component of the boundary, only one of $b_\theta^+$ and $b_\theta^-$ will be nonzero. By \cite[Lem B.7,8]{HT13}, such a positive solution exists. In \cite[Thm 6.1]{HT13}, it was shown that $\beta v$ is a subsolution for $\widehat F$ for $\beta$ sufficiently small. Thus $\beta \psi v$ is a subsolution of the original $F$ by conformal covariance.

The factor $\psi>0$ was independent of $W$, so it is automatically bounded. The size of $\beta$ depended only on the max and min of $v$. Thus to show that $\beta v$ has a lower bound for all admissible $W$, we need only show that $v$ is bounded both above and below independent of $W$.

Our choice of differential operator \eqref{eq:noSubsolution1} on $v$ fulfills the requirements for Lemma \ref{boundedSolution}. Thus
\[
\sup(v) \leq C(\|a_{\widehat w}\|_{p} + \|\widehat b_w + \widehat b_\theta^-\|_{W^{1-1/p,p}(\MN)} + \|\widehat \phi_D\|_{W^{2-1/p,p}(\MD)}).
\] The last two terms are bounded above since they are independent of $W$. For the first term, we calculate
\[
\int_M |a_{\widehat w}|^p \leq C \int_M|\sigma+LW|^{2p} \leq C \int_M |\sigma|^{2p} + |LW|^{2p}.
\] We dropped the hat since the conformal factor $\psi$ has an (uniform) upper bound. We need to bound $|LW|^{2p}$ above for any $W$ that is a solution of the vector problem \eqref{vectorSystem} for some $\phi \in (0,\phi_+]_{2,p}$. The Sobolev embedding $\|LW\|_{\infty} \leq \|LW\|_{1,p}$ combined with Theorem \ref{schauderEstimateTheorem} bounds $|LW|^{2p}$. Thus $v$ has a uniform upper bound.

For the lower bound, by Lemma \ref{boundedSolution},
\[
\inf(v) \geq c_2 \left(\|a_{\widehat w}\|_{L^1(M\setminus N)} + \int_{\MN}(-\widehat b_w - \widehat b_\theta^-)  + \int_{\MD} \widehat \phi_D \right)
\] where $N$ is a neighborhood of the boundary and $c_2$ depends on $N$. If $\MD \neq \emptyset$ or if $\widehat b_w + \widehat b_\theta \not\equiv 0$, this clearly has a uniform lower bound since we can drop the $a_{\widehat w}$ term. We assume otherwise, and thus assume that $\sigma \not\equiv 0$.

We then need to show that $c_2 \int_{M\setminus N} a_w$ has a uniform lower bound. We dropped the hat since $\psi$ has a (uniform) lower bound. Let $N$ be an $\epsilon$ wide neighborhood of $\partial M$. We then let $\epsilon$ be sufficiently small such that
\[
\int_{M\setminus N} |\sigma|^2 \geq \frac{1}{2} \int_M |\sigma|^2.
\] Such an $\epsilon$ must exist or else $\sigma$ would be zero on $M$. We also make $\epsilon$ small enough such that
\[
\int_{\partial (M\setminus N)} \sigma(\nu, W) \geq - \frac{1}{4} \int_M |\sigma|^2.
\] Such an $\epsilon$ must exist since $\sigma\in C^0$, $\sigma(\nu, X_\D) = 0$ on $\WMD$ and $\sigma\cdot \nu = 0$ on $\WMN$, and so the integral on the left goes to zero as $\epsilon \to 0$.

We then have
\begin{align*}
\int_{M\setminus N} a_{ w} &\geq C \int_{M\setminus N} |\sigma +LW|^2\\
 &= C\left(\int_{M\setminus N} (|\sigma|^2 +|LW|^2) - 2\int_{M\setminus N} \di \sigma \cdot W + \int_{\partial(M\setminus N)} \sigma(\nu,W)\right)\\
 &\geq C \int_M |\sigma|^2
\end{align*} and so $v$ has a uniform lower bound. This completes the theorem.

\end{proof}

\begin{thm}\label{noSubsolutionNeededNegative}
Let $\phi_+ \in W^{2,p}$ be a global supersolution. In addition, suppose $b_H \leq \frac{n-2}{2}H$ and $g\in Y^-$. Suppose that there exists a positive solution $u \in W^{s,p}$ of the following problem:
\begin{equation}\label{negSubsolution}
\begin{array}{rlc} -\Delta u + a_Ru +a_\tau u^{N-1}  &\hspace{-2.4mm}= 0\\
\gamma_N\partial_\nu u + b_h u + b_\tau u^{N/2} + b_\theta^+ u^e &\hspace{-2.4mm}= 0 &  \textrm{ on } \MN \\
\gamma_D u &\hspace{-2.4mm}= 1 &  \textrm{ on } \MD. \end{array}
\end{equation} Then there exists $\phi \in W^{2,p}$ and $W\in W^{2,p}$ with $0<\phi\leq \phi_+$ of the combined conformal system \eqref{lichSystem}-\eqref{vectorSystem}.
\end{thm}
\begin{proof}
Note that $u$ does not depend on $W$ or $\phi$. According to the proof of \cite[Thm 6.2]{HT13}, $\beta u$ is a subsolution for small enough $\beta$, and it is easy to see that the $\beta$ does not depend on $W$ or $\phi$. The uniform lower bound on solutions is then $K= \beta \inf u$.
\end{proof}

If $b_h \leq \frac{n-2}{2} H$ and $g\in Y^-$, it was shown in \cite[Thm 6.2]{HT13} that system \eqref{negSubsolution} has a solution if and only if the Lichnerowicz problem \eqref{lichSystem} has a solution. Unfortunately, it is unclear when system \eqref{negSubsolution} has a solution. However, Holst and Tsogtgerel have proven the following partial result.

\begin{lem}\cite[Lem 6.3]{HT13}
Let $g\in Y^-$ and suppose $b_H \leq \frac{n-2}{2}H$. Moreover, assume that there is a constant $c>0$ such that $a_\tau \geq c$ and $b_\tau + b_\theta \geq c $ pointwise almost everywhere. Then there exists a positive solution $u\in W^{2,p}$ to the system \eqref{negSubsolution}.
\end{lem}

\section{Supersolutions}

Theorems \ref{noSubsolutionNeeded} and \ref{noSubsolutionNeededNegative} reduce the problem of finding solutions to the full conformal constraint equations \eqref{lichSystem}-\eqref{vectorSystem} to that of finding global supersolutions. In this section we find several global supersolutions, which are analogous to those found in \cite{HNT09}. Remember that for every supersolution that we find, we then have a solution to the full constraints as long as the Assumptions \ref{assumptions} are fulfilled. Also, though we only consider the vacuum case, these supersolutions are easily adaptable to the scaled energy case, as in \cite{HNT09}.

Let a superscript $\wedge$ will mean the supremum of the function on the appropriate domain, while a superscript $\vee$ will similarly be the infimum. In this section, we'll assume $X = \frac{n-1}{n} d\tau \phi^N$. In this case, using $\|LW\|_\infty \leq C \|LW\|_{1,p}$ and Theorem \ref{schauderEstimateTheorem} we get that
\begin{align}
\|LW\|^2_{\infty} &\leq C\left(\|X\|_{p}^2 + \|X_\N\|^2_{W^{1-\frac{1}{p},p}(\WMN)} + \|X_\D\|^2_{W^{2-\frac{1}{p},p}(\WMD)}\right)\notag\\
&\leq C_1 \|d\tau\|_p^2 (\phi^\wedge)^{2N} +C_2 \label{eq:LWboundSquared}
\end{align} since $X_\N$ and $X_\D$ do not depend on $\phi$.

\begin{thm}[$g\in Y^+$, far-from-CMC] \label{FarFromCMC} Suppose that $g\in Y^+$ and that $b_H \geq \frac{n-2}{2} H$. Suppose that
\begin{equation}\label{eq:quantities}
\begin{cases} \|d\tau\|_p\\
        \|\sigma\|_\infty + \|X_\N \|^2_{W^{1-\frac{1}{p},p}(\WMN)} + \|X_\D\|^2_{W^{2-\frac{1}{p},p}(\WMD)}\\
        \|b_w\|_\infty \\
        \|\phi_D\|_\infty \\
        \|b_\theta^-\|_\infty \end{cases}
\end{equation} are sufficiently small, except perhaps one. Then there exists a global supersolution.
\end{thm}
\begin{proof}
The Yamabe classification Theorem \ref{YamabeClassification} implies that there exist positive functions $u$, $\Lambda_1$ and $\Lambda_2$ such that
\[
\begin{array}{rlc} -\Delta u + a_Ru  &\hspace{-2.4mm}= \Lambda_1\\
\gamma_N\partial_\nu u + \frac{n-2}{2} H u &\hspace{-2.4mm}= \Lambda_2 &  \textrm{ on } \partial M. \end{array}
\] Indeed, $u$ is a conformal factor provided by that theorem that takes $g$ to a metric with positive scalar curvature and positive boundary mean curvature.

Let $\phi_+ = \beta u$. We will set up three expressions that all need to be positive for $\phi_+$ to be a global supersolution. We will then explain why we can pick a $\beta$ to make them all positive. We assume $W$ is admissible for $\phi_+$.

Note that $-\Delta \phi_+ + a_R \phi_+ = \beta \Lambda_1$. We then see that
\begin{align*}
-\Delta \phi_+ &+ a_R\phi_+ + a_\tau \phi_+^{N-1} - a_w \phi_+^{-N-1} \\ &\geq \beta \Lambda_1 + a_\tau (\beta u)^{N-1} - \frac{n-2}{2(n-1)} \left(|\sigma|^2 +|LW|^2\right)(\beta u)^{-N-1} \\
&\geq \beta \Lambda_1 + \left(a_\tau - c_n C_1 \|d\tau\|_p^2 b^{2N}\right) (\beta u)^{N-1} - c_n(|\sigma|^2+C_2)(\beta u)^{-N-1}
\end{align*} where $b = \phi_+^\wedge/\phi_+^\vee = u^\wedge /u^\vee$ and $c_n = \frac{n-2}{2(n-1)}$. We used \eqref{eq:LWboundSquared} for the last line. Thus, for $\phi_+$ to be a supersolution, we need
\begin{equation}\label{eq:FirstSupsolnIneq}
\Lambda_1^\vee - c_n C_1 \|d\tau\|_p^2 b^{2N} \beta^{N-2} (u^\wedge)^{N-1} - c_n((|\sigma|^\wedge)^2+C_2) \beta^{-N-2} (u^\wedge)^{-N-1}\geq 0.
\end{equation}

For the Neumann boundary condition, we similarly need, after dropping the $b_H - \frac{n-2}{2} H$ term,
\begin{equation}\label{eq:SecondSupsolnIneq}
\Lambda_2^\vee - \|b_\theta^-\|_\infty \beta^{e-1}(u^\wedge)^e  - \|b_w\|_\infty \beta^{-N/2-1}(u^\wedge)^{-N/2} \geq 0.
\end{equation} This is because, by Assumption \ref{assumptions}, $b_w\leq 0$. Note that $b_\theta^- \equiv 0$ unless $e-1 < 0$.

For the Dirichlet boundary condition, we need a simpler condition,
\begin{equation}\label{eq:ThirdSupsolnIneq}
\beta u - \phi_D \geq 0.
\end{equation}

Let $\|d\tau\|_p$ be arbitrary, and take $\beta>0$ sufficiently small so that
\[
\Lambda_1^\vee - c_n C_1 \|d\tau\|_p^2 b^{2N} \beta^{N-2} (u^\wedge)^{N-1} > \frac{1}{2}\Lambda_1^\vee >0.
\] Then, if all quantities besides $\|d\tau\|_p$ in \eqref{eq:quantities} are sufficiently small, the desired inequalities \eqref{eq:FirstSupsolnIneq}, \eqref{eq:SecondSupsolnIneq}, and \eqref{eq:ThirdSupsolnIneq} all hold. This establishes the theorem in the case that $\|d\tau\|_p$ is large; the corresponding proof when any other single quantity from \eqref{eq:quantities} is arbitrary is similar, except that we take $\beta$ to be large instead of small.
\end{proof}

This result can be viewed as a far-from-CMC result in two ways. The first is that if we let $\|d\tau\|_p$ be arbitrary, we can clearly construct far-from-CMC solutions to the constraints. On the other hand, the near-CMC assumption is usually of the form $\|d\tau\|_p/|\tau|^\vee$ is sufficiently small. In this theorem, if we take another of the quantities from \eqref{eq:quantities} large, $\|d\tau\|_p$ must be small, but this bound is not dependent on $\tau$. Since $g\in Y^+$, we can make $R>0$, and so $a_R +a_\tau >0$, as required for Theorem \ref{noSubsolutionNeeded}, for any $\tau$. Thus, even if $|\tau|^\vee$ is small or even zero, we can still construct solutions to the conformal constraint equations as long as $d\tau$ is sufficiently small.

The problematic term in the proof of Theorem \ref{FarFromCMC} is the $\|d\tau\|_p$ term. For the rest of the terms, larger $\beta$ makes the desired inequalities \eqref{eq:FirstSupsolnIneq}-\eqref{eq:ThirdSupsolnIneq} more likely to be true. In the proof of the previous theorem we dropped the $a_\tau$ term. If $d\tau$ is sufficiently small, we can use $a_\tau$ to control the $\|d\tau\|_p$ term.

\begin{thm}[$g\in Y^+$, near-CMC]
Suppose that $g\in Y^+$ and that $b_H \geq \frac{n-2}{2} H$. Suppose that $\|d\tau\|_p/ |\tau|^\vee$ is sufficiently small. Then there exists a global supersolution.
\end{thm}
\begin{proof}
We proceed as in Theorem \ref{FarFromCMC} but do not get rid of the $a_\tau$ term. Let $u, \Lambda_1, \Lambda_2\in W^{s,p}$ and $\phi_+$ be as before. Thus, analogously to \eqref{eq:FirstSupsolnIneq}, we need
\begin{equation}\label{eq:nearCMC1}
\Lambda_1 + \left(a_\tau^\vee - c_n C_1 \|d\tau\|_p^2 b^{2N}\right) \beta^{N-2} u^{N-1} - c_n(|\sigma|^2+C_2) \beta^{-N-2} u^{-N-1}\geq 0.
\end{equation} If $\|d\tau\|_p/|\tau|^\vee$ is sufficiently small, the second term of \eqref{eq:nearCMC1} is positive. Thus, \eqref{eq:nearCMC1} is implied by
\begin{equation}\label{eq:nearCMC2}
\Lambda_1 - c_n((|\sigma|^\wedge)^2+C_2) \beta^{-N-2} u^{-N-1}\geq 0.
\end{equation}

The other two conditions are the same, namely,
\begin{equation}\label{eq:nearCMC3}
\Lambda_2^\vee - \|b_\theta^-\|_\infty \beta^{e-1}(u^\wedge)^e  - \|b_w\|_\infty \beta^{-N/2-1}(u^\wedge)^{-N/2} \geq 0
\end{equation}\begin{equation}\label{eq:nearCMC4}
\beta u - \phi_D \geq 0.
\end{equation} All three of \eqref{eq:nearCMC2}-\eqref{eq:nearCMC4} hold for $\beta$ large enough. This completes the proof.
\end{proof}

We can similarly show that there is a global supersolution if $g\in Y^0$, though the proof is a little more complicated.

\begin{thm}[$g\in Y^0$, near-CMC] Suppose that $g\in Y^0$ and that $b_H \geq \frac{n-2}{2}H$. Assume that one of the following holds:
\begin{equation}\label{eq:nearCMCassumptions}
\begin{cases} a_\tau \not\equiv 0\\
b_\theta\leq 0 \textrm{ and } b_\tau \not\equiv 0\\
b_\theta\geq 0 \textrm{ and } b_\tau +b_\theta \not\equiv 0\\
\MD \neq \emptyset
\end{cases}\end{equation} In the first three cases we also assume that either $\sigma$ or $b_\theta^-+b_w$ is not identically zero. Finally, suppose that $\|d\tau\|_p/|\tau|^\vee$ is sufficiently small. Then there exists a global supersolution.
\end{thm}
\begin{proof}
We only consider the case where $b_\theta \leq 0$. The other cases are handled similarly. Let $u,v$ be the solutions of the following equations.
\begin{equation}\label{eq:nearCMCNull1}
\begin{array}{rlc} -\Delta u + a_Ru  &\hspace{-2.4mm}= 0\\
\gamma_N\partial_\nu u + \frac{n-2}{2}H u &\hspace{-2.4mm}= 0 &  \textrm{ on } \partial M \end{array}
\end{equation}\begin{equation}\label{eq:nearCMCNull2}
\begin{array}{rlc} -\nabla(u^2 \nabla v) + a_\tau v  &\hspace{-2.4mm}= c_n |\sigma|^2/2\\
\gamma_N\partial_\nu v + b_\tau v &\hspace{-2.4mm}= -(b_\theta+b_w) &  \textrm{ on } \MN\\
\gamma_D v &\hspace{-2.4mm}= \phi_D &  \textrm{ on } \MD \end{array}
\end{equation} where $c_n = \frac{n-1}{2(n-2)}$ as before.

The Yamabe classification Theorem \ref{YamabeClassification} implies that there exists a positive solution $u\in W^{2,p}$ to \eqref{eq:nearCMCNull1}. Indeed, $u$ is a conformal factor provided by that theorem that takes $g$ to a metric with zero scalar and boundary mean curvatures. A variation of \cite[Lem B.6,7]{HT13} and our assumptions that some of the quantities are not identically zero guarantee that the second system \eqref{eq:nearCMCNull2} has a positive solution $v\in W^{2,p}$. We claim that $\phi_+=\beta uv$ is a global supersolution for sufficiently large $\beta$.

As before, there are three expressions that need to positive in order for $\phi_+$ to be a global supersolution. First note that
\begin{align*}
-u\Delta(\phi_+)+ a_R u \phi_+ &= -\beta u \nabla (v\nabla u + u \nabla v) + \beta uv \Delta u \\
&= -\beta \nabla(u^2 \nabla v) + \beta u\nabla u\nabla v - \beta u \nabla v \nabla u - \beta uv\Delta u + \beta uv\Delta u \\
&= \beta (c_n |\sigma|^2/2 - a_\tau v).
\end{align*}
The first inequality we need is
\begin{equation}\label{eq:nearCMCNull4}
-u\Delta\phi_+ + a_R u\phi_+ + a_\tau u \phi_+^{N-1} - a_w u \phi_+^{-N-1} \geq 0.
\end{equation} We can calculate
\begin{align}
-u\Delta\phi_+ &+ a_R u\phi_+ + a_\tau u \phi_+^{N-1} - a_w u \phi_+^{-N-1}\notag\\ &= \beta (c_n|\sigma|^2/2 - a_\tau v) + a_\tau (\beta v)^{N-1}u^N - a_w (\beta v)^{-N-1} u^{-N}\notag\\
&\geq a_\tau ((\beta v)^{N-1} u^N - \beta v) + \beta c_n|\sigma|^2 - c_n(|\sigma|^2 + |LW|^2)  (\beta v)^{-N-1} u^{-N} \notag\end{align} which simplifies to
\begin{equation}\label{eq:nearCMCNull3}
= a_\tau ((\beta v)^{N-1} u^N - \beta v) - c_n |LW|^2(\beta v)^{-N-1} u^{-N} + c_n |\sigma|^2(\beta/2 - (\beta v)^{-N-1} u^{-N}).
\end{equation} Since $u$ and $v$ are positive, the $|\sigma|^2$ term is positive for large enough $\beta$.

If $W$ is admissible, inequality \eqref{eq:LWboundSquared} holds, i.e., $\|LW\|_\infty^2 \leq C_1 \|d\tau\|_p^2 (\phi_+^\wedge)^{2N} +C_2$. Using this with the first two terms of \eqref{eq:nearCMCNull3}, we get
\begin{multline*}
a_\tau ((\beta v)^{N-1} u^N - \beta v) - c_n |LW|^2(\beta v)^{-N-1} u^{-N} \\ \geq \left[a_\tau^\vee (v^\vee)^{N-1} (u^\vee)^{N}  - C ((uv)^\wedge)^{2N} (u^\vee)^{-N} (v^\vee)^{-N-1} \|d\tau\|_p^2\right] \beta^{N-1} + O(\beta).
\end{multline*} Since $\|d\tau\|_p/|\tau|^\vee$ is sufficiently small, for large enough $\beta$, this quantity is positive. Thus \eqref{eq:nearCMCNull4} holds.

For the Neumann boundary condition, we drop the traces for clarity. We first note that
\[
\partial_\nu(uv) + b_H uv = \left(b_H -\frac{n-2}{2}H\right) uv + u \partial_\nu v
\] and so we can show
\begin{align*}
\partial_\nu \phi_+ + f(\phi_+) &= \left(b_H -\frac{n-2}{2}H\right) \beta uv + \beta u \partial_\nu v + b_\theta \phi_+^e + b_\tau \phi^{N/2}_+ + b_w \phi_+^{-N/2}\\
&\geq -b_\theta (\beta u-\phi_+^e)+b_\tau (\phi_+^{N/2} - \phi_+) - b_w (\beta u - \phi_+^{-N/2}).
\end{align*} Since $b_\theta\leq 0$, $b_\tau \geq 0$ and $b_w\leq 0$ (see Assumptions \ref{assumptions}), this is positive for $\beta$ large enough, and so
\begin{equation}\label{eq:nearCMCNull5}
\partial_\nu \phi_+ + f(\phi_+) \geq 0.
\end{equation}

For the Dirichlet boundary condition, a large $\beta$ clearly gives
\begin{equation}\label{eq:nearCMCNull6}
\gamma_D \phi_+ -\phi_D>0.
\end{equation} The inequalities \eqref{eq:nearCMCNull4}, \eqref{eq:nearCMCNull5} and \eqref{eq:nearCMCNull6} together show that $\phi_+$ is a global supersolution for large enough $\beta$.
\end{proof}

\begin{thm}[$g\in Y^-$, near-CMC] Assume the conditions of Theorem \ref{noSubsolutionNeededNegative} are met, except for the existence of a global supersolution. Suppose that either $\sigma \not\equiv 0$, $b_w+b_\theta^- \not\equiv 0$ or that $\MD \neq \emptyset$. Finally, suppose that $\|d\tau\|_p/|\tau|^\vee$ is sufficiently small. Then there exists a global supersolution.
\end{thm}
\begin{proof}
We only consider the case where $b_\theta\geq 0$. The other case is handled similarly. Let $u$ be the solution to \eqref{negSubsolution} from Theorem \ref{noSubsolutionNeededNegative}. The function $u$ is a conformal factor that transforms $g$ to a metric with scalar curvature to $-a_\tau$. Since $g\in Y^-$, $-a_\tau$ cannot be identically zero. After the conformal transformation by $u$, the Lichnerowicz problem \eqref{lichSystem} reads
\begin{equation}\label{eq:noSubsolutionNeg1}
\begin{array}{rlc} -\Delta\phi -a_\tau \phi + a_\tau \phi^{N-1} - a_w \phi^{-N-1} &\hspace{-2.4mm}= 0\\
\gamma_N\partial_\nu\phi -(b_\tau + b_\theta u^{e-\frac{N}{2}}) \phi + b_\theta \phi^e + b_\tau \phi^{N/2} + b_w \phi^{-N/2}  &\hspace{-2.4mm}=0 &  \textrm{ on } \MN\\
\gamma_D\phi-\phi_D & \hspace{-2.4mm}= 0  &  \textrm{ on } \MD.\end{array}
\end{equation} Let $v\in W^{2,p}$ be the solution to
\begin{equation}\label{eq:noSubsolutionNeg2}
\begin{array}{rlc} -\Delta v + a_\tau v &\hspace{-2.4mm}= c_n |\sigma|^2/2\\
\gamma_N\partial_\nu v +(b_\tau + b_\theta u^{e-\frac{N}{2}}) v  &\hspace{-2.4mm}= -b_w &  \textrm{ on } \MN\\
\gamma_D\phi-\phi_D & \hspace{-2.4mm}= 0  &  \textrm{ on } \MD.\end{array}
\end{equation} The condition $a_\tau \not\equiv 0$ guarantees that there is a unique solution $v$ to \eqref{eq:noSubsolutionNeg2}. The assumption that either $\sigma$ or $b_w$ are nonzero or that $\MD$ is nonempty gives that $v>0$. One can show that $\phi_+ = \beta v$ is a supersolution for \eqref{eq:noSubsolutionNeg1} for sufficiently large $\beta>0$, as in the previous theorem, under the near-CMC assumption given. Since $v$ does not depend on $W$, this is a global supersolution.
\end{proof}

\section{``Limit Equation" Results and Inequalities}\label{sec:LimitEquation}

In the papers \cite{DGH11,GS12,DGI13} it has been shown that there is a ``limit equation," such that either it or the constraint equations has a solution (or both). As part of the proof they prove several independently useful existence and inequality results that are not clear from their presentation. For instance, in the closed manifold case (in \cite{DGH11}) they prove that $\|\phi^N\|_\infty \leq C \max\{1,\|LW\|_2\}$ for any solution of the constraint equations. This is the opposite direction of the more easily shown inequality $\|LW\|_2 \leq C\|\phi^N\|_\infty$ that is often used.

In our case, the compact with boundary case, it proves difficult to make the last step in order to prove the existence of a solution to the limit equation. However, all of the other results have analogues. Since they may be of independent value, we prove them here.

We assume that $X = \frac{n-1}{n} \phi^{N-\epsilon} d\tau$ for some $\epsilon\in [0,1)$, though we could include a scaled energy term without much difficulty. We also need slightly more regularity for $X_\N$ and $X_\D$, namely we need $X_\N\in {W^{1-\frac{3}{5n},\frac{5n}{3}}(\WMN)}$ and $X_\D\in {W^{2-\frac{3}{5n},\frac{5n}{3}}(\WMD)}$. This may be already satisfied because of our standard assumptions \eqref{assumptions}, depending on our choice of $p$.

Also, we need that $F_2(\Lambda) \geq 0$ for any large constant $\Lambda$, where $F_2$ is the line of the Lichnerowicz problem \eqref{lichSystem}. We assume that this is true. Note that this happens, in particular, in the defocusing case when the $b$ coefficients do not depend on $W$ and particular $b$ coefficients are non-zero. It would be sufficient for the coefficient of the highest power of $\phi$ in $F_2$ to be strictly positive, though that is slightly stronger than we require.

Finally, we require $\inf\tau>0$, where we assume $\tau>0$ rather than $\tau<0$ without loss of generality. This is similar to \cite{DGH11, GS12, DGI13}.

If $\epsilon\neq 0$, we will refer to the conformal constraint equations with these $X$'s as the (conformal) constraint equations with $\epsilon$.

In this section we will prove the following three lemmas.

\begin{lem}\label{subcriticalExistence}
Suppose the conditions of either Theorem \ref{noSubsolutionNeeded} or \ref{noSubsolutionNeededNegative} hold, in both cases except for the existence of a global supersolution. Also suppose that $\epsilon>0$. Then there exists solutions $\phi, W\in W^{2,p}$ to the conformal constraint equations with $\epsilon$.
\end{lem}

\begin{lem} \label{subcriticalInequality}
Suppose $\phi, W\in W^{2,p}$ are solutions of the conformal constraint equations with $\epsilon \in [0,1)$ under the same conditions as Lemma \ref{subcriticalExistence}. Also suppose $g\in W^{2,q}$, $q \geq \frac{n}{2}\left(2+\frac{np}{p-n}\right)$ (or just $g\in C^2$). Then the following inequality holds, with $C$ independent of $\phi$, $W$ and $\epsilon$:
\begin{align}\label{eq:subcriticalInequality}
\|\phi^{2N}\|_\infty &\leq C \widetilde\gamma
\end{align} where $\widetilde\gamma$ is a constant defined below depending on $\|LW\|_2$ and the boundary values of $\phi$.
\end{lem}

\begin{lem}\label{criticalExistence}
Suppose the same conditions as for Lemma \ref{subcriticalInequality} are fulfilled. Let $\epsilon_i$ and $(\phi_i, W_i)$ be a sequence such that $\epsilon_i\geq0$, $\epsilon_i\to 0$ and $(\phi_i, W_i)$ is a solution of the conformal equations with $\epsilon = \epsilon_i$. Also assume that the conditions of Lemma \ref{LichnerowiczContinuity} (the continuity of the Lichnerowicz problem) are fulfilled. If the right side of the inequality \eqref{eq:subcriticalInequality} is uniformly bounded then there exists a subsequence of the $(\phi_i, W_i)$ which converges in $W^{2,p}$ to a solution $(\phi_\infty,W_\infty)$ of the original conformal constraint equations.
\end{lem}

The limit equation appears by considering what happens when $\widetilde\gamma$ is unbounded. However, there are some difficulties that appear in this case that do not appear in other cases, which we will discuss below.

We first prove Lemma \ref{subcriticalExistence}.

\begin{proof}[Proof of Lemma \ref{subcriticalExistence}]
By Theorem \ref{noSubsolutionNeeded} or \ref{noSubsolutionNeededNegative}, all we need to find is a (global) supersolution. We claim there is a constant supersolution. Let $W$ be admissible for the possible constant supersolution $\Lambda$. We want to show that $\Lambda$ is a supersolution to the Lichnerowicz problem for any such $W$ for $\Lambda$ large enough.

First, using Theorem \ref{schauderEstimateTheorem},
\begin{align}\label{eq:subcriticalExistence1}
\|LW\|_{\infty} \leq  C (\Lambda^{N-\epsilon} + 1).
\end{align}

By Theorem \ref{YamabeClassification}, use a conformal transformation to change $g$ to a metric with continuous (and thus bounded) scalar curvature $R$. Using this, we get, where $F_1(\phi)$ is the first line of the Lichnerowicz problem \eqref{lichSystem},
\begin{align*}
F_1(\Lambda) &= a_R\Lambda + a_\tau \Lambda^{N-1} - a_w \Lambda^{-N-1} \\
&\geq C_1\Lambda + C_2 \Lambda^{N-1} - (C_3 |\sigma|^2+C_4) \Lambda^{-N-1} - C_5 \Lambda^{N-1- 2\epsilon}
\end{align*} for constant $C_1$ and positive constants $C_2, C_3, C_4$ and $C_5$. Thus for large enough $\Lambda$ and $\epsilon>0$, $F_1(\Lambda)>0$.

For $F_2$, the second line of the Lichnerowicz problem \eqref{lichSystem},  $F_2(\Lambda)>0$ by assumption, as discussed in the introduction for this section. Clearly $F_3(\Lambda)>0$, where $F_3$ is the third line of the Lichnerowicz problem \eqref{lichSystem}.

Combining these gives that $F(\Lambda)>0$ for large enough $\Lambda$, and so $\Lambda$ is a global supersolution. By Theorem \ref{noSubsolutionNeeded} (or \ref{noSubsolutionNeededNegative}), there is a solution $(\phi_\epsilon,W_\epsilon) \in W^{2,p}\times W^{2,p}$ to the conformal constraint equations with $\epsilon$.
\end{proof}

\subsection{Convergence of subcritical solutions.}

Let $1>\epsilon\geq 0$ and let $(\phi,W)$ be the solution found previously. We define an energy of this solution as
\[
\gamma(\phi,W) := \int_M |LW|^2 + \sup_{\partial M}\left(  \phi^{N+1+\frac{np}{p-n}}|\partial_\nu\phi|\right)
\] and set $\widetilde \gamma = \max\{\gamma,1\}$. We want to show that $\phi$ has an upper bound depending only on $\widetilde \gamma$ but otherwise independent of $\phi$ or $W$. Note that we allow $\epsilon =0$ here.

To do this, we transform the conformal equations by $\widetilde \gamma$. Since for this section we won't need the boundary conditions, we will not write the boundary equations. We rescale $\phi$, $W$ and $\sigma$ as
\begin{equation}\label{eq:subcriticalRescaling}
\widetilde\phi = \widetilde\gamma^{-\frac{1}{2N}} \phi, \,\,\,\, \widetilde W = \widetilde\gamma^{-\frac{1}{2}} W, \,\,\,\, \widetilde\sigma = \widetilde\gamma^{-\frac{1}{2}} \sigma.
\end{equation} The conformal constraint equations with $\epsilon$ can then be renormalized as
\begin{equation}\label{newLich}
\frac{1}{\widetilde\gamma^{1/n}} \left(\frac{-4(n-1)}{n-2}\Delta \widetilde\phi + R\widetilde\phi\right) + \frac{n-1}{n} \tau^2 \widetilde\phi^{N-1} = |\widetilde\sigma+L\widetilde W|^2 \widetilde\phi^{-N-1},
\end{equation}\begin{equation}\label{newVect}
\di L\widetilde W = \frac{n-1}{n} \widetilde\gamma^{-\frac{\epsilon}{2N}}\widetilde\phi^{N-\epsilon} d\tau
\end{equation} by dividing through by an appropriate power of $\widetilde\gamma$. Notice that because of our rescaling,
\[
\int_M |L\widetilde W|^2\, dv \leq 1
\] and
\begin{equation}\label{eq:energyBound}
\frac{1}{\widetilde\gamma^{1/n}}\int_{\partial M} \widetilde{\phi}^{N+1+k} \partial_\nu \widetilde\phi = \frac1{\widetilde\gamma} \int_{\partial M} \phi^{N+1} \partial_\nu \phi \geq -C
\end{equation} for any $0\leq k<\frac{np}{p-n}$ and some constant $C$ independent of $\phi$. The choice of the boundary terms in our energy $\gamma$ was to bound this second quantity \eqref{eq:energyBound}. In order to prove Lemma \ref{subcriticalInequality}, we first need a lemma.

\begin{lem} \label{bound}
Suppose  $g\in W^{2,q}$, $q \geq \frac{n}{2}\left(2+\frac{np}{p-n}\right)$ (or just $g\in C^2$). Then, for any $0 \leq k_i < \frac{np}{p-n}$, the following inequality holds, with $C_i,C>0$ independent of $\epsilon, \phi$ and $W$.
\[
-{C_i} \left( \int_M \widetilde\phi^{2N+Nk_i} dv\right)^{\frac{N+2+Nk_i}{2N+Nk_i}} + \tau_0^2\int_M  \widetilde \phi^{2N+Nk_i} \leq C + \int_M |\widetilde \sigma + L\widetilde W|^2 \widetilde\phi^{Nk_i}
\]
\end{lem}
\begin{proof}
We multiply equation (\ref{newLich}) by $\widetilde\phi^{N+1+Nk_i}$ and integrate over $M$ to get
\begin{multline*}
\frac{1}{\widetilde\gamma^{1/n}}\int_M \left(-c_n\widetilde\phi^{N+1+Nk_i}\Delta \widetilde\phi + R \widetilde\phi^{N+2+Nk_i}\right) dv \\+ \frac{n-1}{n}\int_M \tau^2 \widetilde\phi^{2N+Nk_i} dv = \int_M |\widetilde\sigma + L \widetilde W|^2 \widetilde\phi^{Nk_i}dv
\end{multline*} where $c_n = \frac{4(n-1)}{n-2}$. After integrating by parts, we get
\begin{multline}\label{eq:subcriticalInequality1}
\frac{1}{\widetilde\gamma^{1/n}}\int_M c_n(N+1+Nk_i) \widetilde\phi^{N+Nk_i}|d\widetilde\varphi|^2 dv - \frac{c_n}{\widetilde\gamma^{1/n}}\int_{\partial M} \widetilde{\phi}^{N+1+k_i} \partial_\nu \widetilde\phi dv \\+ \frac{1}{\widetilde\gamma^{1/n}}\int_M R \widetilde\phi^{N+2+Nk_i} dv + \frac{n-1}{n}\int_M \tau^2 \widetilde\phi^{2N+Nk_i} dv \leq \int_M |\widetilde\sigma + L \widetilde W|^2 \widetilde\phi^{Nk_i}dv.
\end{multline} Since the first integral in \eqref{eq:subcriticalInequality1} is positive, we can get rid of it. The second integral is greater than -C by our choice of $\widetilde \gamma$, and so we can replace it by -C and then add it to the right hand side. We use H\"older's inequality on the third integral, with the exponent on $R$ being $\frac{n}{2}(2+k_i)$. Our assumptions on $g$ and $k_i$ guarantee that this integral is finite. Thus we get
\[
-\frac{C_i}{\widetilde\gamma^{1/n}} \left( \int_M \widetilde\phi^{2N+Nk_i} dv\right)^{\frac{N+2+Nk_i}{2N+Nk_i}} + \tau_0^2\int_M  \widetilde \phi^{2N+Nk_i} \leq C + \int_M |\widetilde \sigma + L\widetilde W|^2 \widetilde\phi^{Nk_i}
\] where $C_i$ are some constants depending on $k_i$ and $R$ and where $\tau_0$ is the (positive) infimum of $\tau$. Using $\widetilde\gamma \geq 1$, we get the desired inequality.
\end{proof}

\begin{prop}
For $1>\epsilon\geq 0$,
\[
\phi < C \widetilde\gamma^{1/2N}
\] for some constant $C$ independent of $\epsilon, \phi$ and $W$.
\end{prop}

Note that this implies Lemma \ref{subcriticalInequality}.

\begin{proof}
For this proof, ``bounded'' will mean bounded independent of $\epsilon, \phi$ and $W$.

\textbf{Step 1.} \emph{$L^1$ bound on $\widetilde\phi^{2N}$.} Using Lemma \ref{bound} with $k_i = 0$,
\begin{align*}
-{C_i} \left( \int_M \widetilde\phi^{2N} dv\right)^{\frac{N+2}{2N}} + \tau_0^2 \int \widetilde\phi^{2N} &\leq C+ 2\int |\widetilde\sigma|^2\, dv +2\int_M |L\widetilde W|^2\, dv\\
&\leq C+2 + 2\int_M |\widetilde\sigma|^2\, dv.
\end{align*} By the rescaling \eqref{eq:subcriticalRescaling} of $\widetilde \sigma$ and recalling that $\widetilde\gamma\geq 1$, $\int |\widetilde\sigma|^2 \, dv$ is bounded. Since $\frac{N+2}{2N} = \frac{n-1}{n} <1$, this implies the $L^1$ bound on $\widetilde\phi^{2N}$.

\textbf{Step 2.} \emph{Bounds for $LW$.} Suppose by induction $\widetilde\phi^{p_iN}$ is bounded in $L^1$ for some $p_i\geq 2$. Let $\frac{1}{q_i} = \frac{1}{p_i} + \frac{1}{p}$ and $\frac{1}{r_i} = \frac{1}{q_i} - \frac{1}{n}$. If $q_i>n$ we continue on to step 4. We can make it so that $q_i$ is never exactly $n$, as argued at the end of step 3.

Young's inequality gives us
\[
\widetilde\phi^{N-\epsilon} \leq \frac{N-\epsilon}{N} \widetilde\phi^N + \frac{\epsilon}{N}
\] and so
\[
\|\widetilde\phi^{N-\epsilon}\|_{p_i} \leq \frac{N-\epsilon}{N} \|\widetilde\phi^N\|_{p_i} + \frac{\epsilon}{N} \vol(M)^{\frac{1}{p_i}} \leq \|\widetilde\phi^N\|_{p_i} + \frac{1}{N} \max\{1, \vol(M)\}.
\]

Using the rescaled vector equation \eqref{newVect} we get
\begin{align*}
\| \di L\widetilde W\|_{q_i} &\leq C \|\widetilde\phi^{N-\epsilon} d\tau\|_{{q_i}} \\
&\leq C \| \widetilde\phi^{N-\epsilon}\|_{p_i} \|d\tau\|_{p}\\
&\leq C \left(\|\widetilde\phi^{p_i N}\|_{1}^{1/p_i} + \frac{1}{N} \max\{1, \vol(M)\}\right)\|d\tau\|_{p}.
\end{align*} The second line is H\"older's inequality with $p_i$ and $p$.

For $q_i<n$, we then get
\begin{equation} \label{LWbound}
\|L\widetilde W\|_{{r_i}} \leq C\|\widetilde W\|_{2,q_i} \leq C\left(\| \di L\widetilde W\|_{{q_i}} + \|\widetilde X_\N \|_{W^{1-\frac{1}{q_i},q_i}(\WMN)} + \|\widetilde X_\D\|_{W^{2-\frac{1}{q_i},q_i}(\WMD)}\right)
\end{equation} where $C$ changes from term to term. This is a variation of Theorem \ref{schauderEstimateTheorem}. The first inequality is by Sobolev embedding since $q_i<n$. The first term is bounded by the previous set of inequalities. As shown at the end of step 3, $q_i <5n/3$, and so $W^{1-\frac{3}{5n},\frac{5n}{3}} \subset W^{1-\frac{1}{q_i},q_i}$ by Sobolev embedding. Thus $\|X_\N\|$ and $\|X_\D\|$ are bounded.

Thus $\|L\widetilde W\|_{r_i}$ is bounded.

\textbf{Step 3.} \emph{Induction on $p_i$.} By Lemma \ref{bound}, $\widetilde\phi^{2N+Nk_i}$ is bounded in $L^1$ as long as
\[
\int_M (|\widetilde \sigma|^2 + |L\widetilde W|^2) \widetilde\phi^{Nk_i}
\] is bounded. Choose $k_i$ by $\frac{2}{r_i} + \frac{k_i}{p_i} = 1$. Using H\"older's inequality with exponents $\frac{r_i}{2}$ and $\frac{p_i}{k_i}$, we get
\[
\int_M (|\widetilde \sigma|^2 + |L\widetilde W|^2) \widetilde\phi^{Nk_i} \leq (\|\widetilde\sigma\|_{{r_i}} + \|L\widetilde W \|_{{r_i}}) \|\widetilde\phi^{p_i N}\|_{1}^{1/p_i}.
\] Since $\sigma \in W^{1,p} \subset L^\infty$, by our induction assumption and by step 2, this quantity is bounded.

Thus $\widetilde \phi^{2N+Nk_i}$ is bounded in $L^1$. Let $p_{i+1} = 2+k_i$. We see that
\[
\frac{p_{i+1}}{p_i} = 1 + 2\left(\frac{1}{n} - \frac{1}{p}\right) >1
\] and so $p_i \to \infty$. Since $p>n$, there is an $i_0$ such that $q_{i_0}\geq n$ and $q_{i_0-1} < n$. If $q_i = n$, we reduce the power $p_i$ somewhat to prevent this, since $\widetilde\phi^{p_i N}$ will still be bounded in $L^1$. If $q_i >n$, we continue to step 4. Note that this definition of $p_{i+1}$ guarantees that $q_{i_0}$, the first $q>n$, is less than $\frac{5n}{3}$. This can be seen by a straightforward calculation that we omit for brevity.

\textbf{Step 4.} \emph{Finishing.} Since $q_i >n$, similar to equation \eqref{LWbound} from step 2,
\[
\|L\widetilde W\|_{\infty} \leq C\|\widetilde W\|_{{2,q_i}} \leq C\left(\| \di LW\|_{{q_i}} + \|X_\N \|_{W^{1-\frac{1}{q_i},q_i}(\WMN)} + \|X_\D\|_{W^{2-\frac{1}{q_i},q_i}(\WMD)}\right)
\] which is bounded as before. Thus $|L\widetilde W|$ has an upper bound.

From standard elliptic regularity, $\widetilde \phi$ is in $C^{1}\supset W^{2,p}$, since all the coefficients in the rescaled Lichnerowicz equation \eqref{newLich} are at least $L^{p}$. Suppose that $\widetilde \phi$ has an internal maximum at some point $x\in M$. At such a point,
\[
\frac{1}{\widetilde \gamma^{1/n}} R \widetilde\phi + \frac{n-1}{n} \tau^2 \widetilde\phi^{N-1} \leq |\widetilde\sigma + L\widetilde W|^2 \widetilde\phi^{-N-1}
\] which simplifies to
\begin{equation}\label{phiBoundInequality}
\frac{1}{\widetilde \gamma^{1/n}} R \widetilde\phi^{N+2} + \frac{n-1}{n} \tau^2 \widetilde\phi^{2N} \leq |\widetilde\sigma + L\widetilde W|^2.
\end{equation} After a conformal change to make $R$ continuous, we see that $\widetilde \phi$ must be bounded.

However, there still could be a larger value on the boundary $\partial M$. If $\sup_M \widetilde\phi$ is located on $\MD$, $\widetilde \phi$ is bounded, since $\phi_D$ is independent of $\phi$ and $W$ and is thus bounded. If the maximum is located on $\MN$, we need to show $\widetilde\phi$ is bounded there too.

Since $\widetilde\phi \in W^{2,p}$, $\gamma_N \partial_\nu \widetilde\phi \in C^0(\MN)$. Also, since $\MN$ is a closed manifold, $\gamma_N\widetilde\phi$ has a maximum on $\MN$. We drop the $\gamma_N$ for the rest of this discussion. Suppose that the maximum of $\widetilde\phi$ is at $x\in\MN$ and is larger than the bound implied by the inequality \eqref{phiBoundInequality}. Then $\Delta \widetilde \phi>0$ in some neighborhood of $x$. As discussed in the proof of Theorem \ref{GreensExistence}, the Hopf lemma applies. In particular, we get $\partial_\nu \widetilde\phi(x) >0$ and thus $\partial_\nu \phi(x) >0$.

Since $\partial_\nu\phi + \widetilde f(\phi) = 0$, we see that
\[
b_H \phi +  b_\theta \phi^e + b_\tau \phi^{N/2} + b_w \phi^{-N/2} <0
\] at $x$. However, this sets a different upper bound on $\phi$ by our assumption in the introduction of this section that $F_2(\Lambda)>0$ for large enough constants $\Lambda$.

By recalling that $\widetilde\phi = \widetilde\gamma^{-\frac{1}{2N}} \phi$, we have proven the proposition.
\end{proof}

Now that we have an upper bound on solutions, we will consider what happens as $\epsilon\to 0$.

\begin{proof}[Proof of Lemma \ref{criticalExistence}]
From Lemma \ref{subcriticalInequality}, we know that the $\phi_i$ are uniformly bounded in $L^\infty(\overline{M})$. From the vector problem \eqref{vectorSystem}, the sequence $W_i$ is uniformly bounded in $W^{2,p}$, after using
\[
\phi^{N-\epsilon} \leq \frac{N-\epsilon}{N} \phi^N + \frac{\epsilon}{N} \leq \phi^N + \frac{1}{N}
\] if $\epsilon>0$. By Sobolev embeddings, the map $L:W^{2,p} \to L^\infty$ is compact. Thus, up to selecting a subsequence, we can assume that the sequence $LW_i$ converges in $L^q$ for any $q\geq 1$ to some $LW_\infty$. Thus by the continuity of the solution map for the Lichnerowicz problem (Lemma \ref{LichnerowiczContinuity}), the functions $\phi_i$ converge in $W^{2,p}$ (and thus in $L^\infty$) to some $\phi_\infty$. Then using the vector problem again, we get that the sequence $W_i$ converges in the $W^{2,p}$ norm. Since $(\phi_i, W_i)$ converge in $W^{2,p}$, $(\phi_\infty, W_\infty)$ are solutions to the conformal constraint equations \eqref{lichSystem}-\eqref{vectorSystem}. Note that convergence in $W^{2,p}$ in the interior gives the appropriate convergence on the boundary since, for instance, $\|\gamma_D\phi\|_{W^{2-1/p,p}(\MD)} \leq C\|\phi\|_{{2,p}}$. Thus $\phi_\infty, W_\infty$ also fulfill the boundary conditions.
\end{proof}

In other cases (cf. \cite{DGH11}) assuming that the energies $\widetilde\gamma$ are unbounded leads to a nontrivial solution of a PDE called the ``limit equation.'' Thus, either this PDE has a solution or the conformal constraint equations do. However, finding the solution to the limit equation is harder in the compact with boundary case. Finding the solution relies on finding a sub/supersolution to the modified Lichnerowicz equation \eqref{newLich}. While the proof that it is a sub/supersolution in the interior of $M$ goes through exactly the same, the same type of argument on $\partial M$ does not work. On the Dirichlet portion $\MD$, for instance, $\widetilde\phi = \widetilde\phi_D \to 0$ as the energy goes to infinity. Thus any subsolution must be non-positive, which makes the subsolution we would normally take not work. Similar problems occur on Neumann part $\MN$. We were not able to resolve these difficulties. However, the other results may prove useful, and so we included this section in the paper.

\section{Acknowledgements}

We would like to thank James Isenberg for useful discussions on the subject. This research was partially supported by the NSF grant DMS-1263431. This material is based upon work supported by the National Science Foundation under Grant No. 0932078 000, while the author was in residence at the Mathematical Sciences Research Institute in Berkeley, California, during the Fall of 2013.

\appendix

\section{Multiplication, Composition and a Green's Function}

\begin{lem} \label{multLemma} \cite[Lem 28]{HNT09}
Let $s_i\geq s$ with $s_1 + s_2 \geq 0$, and $1 \leq p, p_i \leq \infty$ ($i = 1,2$) be real numbers satisfying
\[
s_i - s \geq n \left(\frac{1}{p_i} - \frac{1}{p}\right), \,\,\,\,\, s_1 + s_2 - s > n \left( \frac{1}{p_1} + \frac{1}{p_2} - \frac{1}{p}\right),
\] where the strictness of the inequalities can be interchanged if $s \in \N_0$. In case $\min(s_1,s_2)<0$, in addition let $1<p,p_i < \infty$, and let
\[
s_1 + s_2 \geq n \left( \frac{1}{p_1}+ \frac{1}{p_2} -1\right).
\] Then, the pointwise multiplication of functions extends uniquely to a continuous (and thus bounded for $s_i,s\geq 0$) bilinear map
\[
W^{s_1,p_1}(M) \otimes W^{s_2,p_2}(M) \to W^{s,p}(M).
\]
\end{lem}

\begin{cor} \label{boundedMult}
If $p >1$ and $s>n/p$, then $W^{s,p}$ is a Banach algebra. Moreover, if in addition $q>1$ and $k\in [-s,s]$ satisfy $k-\frac{n}{q} \in [-n-s+\frac{n}{p}, s- \frac{n}{p}]$, then
\[
\|fg\|_{k,q} \leq C \|f\|_{k,q} \|g\|_{s,p}
\] for any $f\in W^{k,q}$, $g\in W^{s,p}$ and some constant $C$ independent of $f$ and $g$.
\end{cor}

The following lemma seems like it should be well known, but we couldn't find a reference, so we include a proof.

\begin{lem}\label{compositionLemma}
Suppose $u\in W^{s,p}$ with $s>n/p$. Let $m = \lceil s\rceil$, and $f\in C^m(I)$ while all its derivatives are in $L^\infty(I)$ where $I$ is the (possibly infinite) range of $u$. Then $f\circ u \in W^{s,p}$ and
\[
\|f\circ u\|_{s,p} \leq \sum_{i=0}^m C_i \|u\|_{s,p}^i
\] for $C_i$ independent of $u$. If $\|u\|_\infty\geq \epsilon>0$, then we can set $C_0 = 0$, and have $C_i$ independent of of $u$, for such $u$.
\end{lem}
\begin{proof}
We first assume $s=3$.

First, $\|f(u)\|_p$ is clearly bounded by a constant. If $u\geq \epsilon$, we can set $\|f(u)\|_p \leq C \|u\|_\infty \leq C\|u\|_{3,p}$.

Next, we see
\[
\|\nabla f(u)\|_p = \|f'(u) \nabla u \|_p \leq \sup |f'| \|\nabla u\|_p \leq C \|u\|_{3,p}.
\]

Next,
\begin{align*}
\|\nabla^2 f(u)\|_p &= \|f'(u) \nabla^2 u + f''(u) |\nabla u|^2\|_p\\
&\leq C(\|f'(u) \nabla^2 u\|_p + \|f''(u) |\nabla u|^2\|_p)\\
&\leq C(\| \nabla^2 u \|_p + \|\nabla u\|_{2p}^2)\\
&\leq C(\|u\|_{3,p} + \|u\|_{1,2p}^2).
\end{align*} We thus need $\|u\|_{1,2p} \leq C\|u\|_{3,p}$.  Sobolev embedding tells us we need
\[
\frac{1}{2p} \geq \frac{1}{p} - \frac{2}{n}
\] which is true since $p> n/3$.

Finally,
\begin{align*}
\|\nabla^3 f(u)\|_p &= \|f'(u) \nabla^3 u + 2 f''(u) \nabla^2 u \nabla u + f'''(u) (\nabla u)^3\|_p \\
&\leq C(\|u\|_{3,p} + \|\nabla^2 u \nabla u\|_p + \|\nabla u\|_{3p}^3)\\
&\leq C(\|u\|_{3,p} + \|\nabla^2 u\|_{3p/2} \|\nabla u\|_{3p} + \|u\|^3_{1,3p})\\
&\leq C(\|u\|_{3,p} + \|u\|_{3,p} \|u\|_{3,p} + \|u\|_{3,p}^3 ).
\end{align*} The third line is by H\"older's inequality. The last line follows from Sobolev embedding, as before. This proves the lemma for $s=3$. If $s$ were another positive integer, the result can be proven similarly, though with more combinatorial complexity.

Next, assume $s = 2+\sigma$ with $\sigma \in (0,1)$. By the definition of these spaces (cf. \cite[Def A.1]{HT13}) we only need to show
\[
\|\nabla^2 f(u)\|_{\sigma,p} \leq \sum_{i=1}^m C_i \|u\|_{s,p}^m.
\] We calculate
\begin{align*}
\|\nabla^2 f(u)\|_{\sigma,p} &\leq C(\|f'(u) \nabla^2 u\|_{\sigma,p} + \|f''(u) |\nabla u|^2\|_{\sigma,p}.
\end{align*} Since $u \in W^{2+\sigma, p}$, $u \in W^{1,q_1}\cap W^{2,q_2}$ where
\[
q_1 = \frac{np}{n-p(1+\sigma)} \hspace{10mm} q_2 = \frac{np}{n-p\sigma}.
\] Since $f\in C^3$, our previous work implies that $f'(u) \in W^{2,q_2}$ and $f''(u) \in W^{1,q_1}$. Lemma \ref{multLemma} then shows that
\begin{align*}
\|\nabla^2 f(u)\|_{\sigma,p} &\leq C(\|f'(u)\|_{2,q_2} \| u\|_{s,p} + \|f''(u)\|_{1,q_1}\|u\|_{1+\sigma,2p}^2\\
&\leq C\|f\|_{C^3}( \|u\|_{s,p} + \|u\|_{s,p}^2)
\end{align*} where $\|u\|^2_{1+\sigma, 2p} \leq C\|u\|^2_{s,p}$ by Sobolev embedding, as before. The result can be proved for any other $s$ similarly, though, again, with more combinatorial complexity.
\end{proof}

\begin{cor}\label{compositionCorollary}
Suppose $u_1,u_2\in W^{s,p}$ with $s>n/p$. Let $m = \lceil s\rceil$, and $f\in C^m(I)$ while all its derivatives are in $L^\infty(I)$ where $I$ is the (possibly infinite) range of $u$. Then $f(u_1) - f(u_2) \in W^{s,p}$ and
\[
\|f(u_1) - f(u_2)\|_{s,p} \leq \sum_{i=0}^m C_i\|u_1-u_2\|_{s,p}^i
\] for $C_i$ independent of $u_i$. If $\|u_1-u_2\|_\infty\geq \epsilon>0$, then we can set $C_0 = 0$, and have $C_i$ independent of of $u_i$, for any such pair $u_i$.
\end{cor}

\begin{cor}\label{NormSplitting}
Suppose $u\in W^{s,p}$ with $s>n/p$. Let $m = \lceil s\rceil$, and $f\in C^m$ while all its derivatives are in $L^\infty(I)$ where $I$ is the (possibly infinite) range of $u$. Also, let $v\in W^{\sigma,q}$, where $q>1$ and $\sigma \in [-s,s]\cap [-n-s+\frac{n}{p} + \frac{n}{q}, s-\frac{n}{p} +\frac{n}{q}]$. Then $v \cdot f(u) \in W^{\sigma,q}$ and
\[
\|v\cdot f(u)\|_{\sigma,q} \leq C \|v\|_{\sigma,q} \sum_{i=0}^m C_i \|u\|_{s,p}^i.
\] If $u\geq \epsilon>0$, then we can set $C_0 = 0$.

We can also modify this theorem in a similar way as Corollary \ref{compositionCorollary}.
\end{cor}
\begin{proof}
This follows immediately from Lemma \ref{compositionLemma} and Corollary \ref{boundedMult}.
\end{proof}

Next we will show the existence of the Green's function for the operator
\[
Lu = \left\{ \begin{array}{cl} -\Delta u + \alpha u & \textrm{   on } M \\
                                \partial_\nu u + \beta u & \textrm{   on } \MN \\
                                u & \textrm{   on } \MD. \\ \end{array} \right.
\]  We follow \cite{Aubin98}.

\begin{thm} \label{GreensExistence}
Let $(M^n,g)$ be a smooth compact manifold with boundary with $g\in W^{2,p}$, where $p>n$ and $n\geq 3$. Let $\alpha \in L^\infty(M)$ with $\alpha\geq 0$ and $\beta \in W^{1-\frac{1}{p},p}(\MN)$ with $\beta\geq 0$. Assume also that either $\alpha \not\equiv 0$, $\beta \not\equiv 0$ or $\MD \neq \emptyset$. Then there exists a Green's function $G(x,y)$ for the operator $L$ with the following properties:
\begin{enumerate}[(a)]

\item $G(x,y) = 0$ for $y \in \MD$ and $\partial_\nu G(x,y) + \beta G(x,y) = 0$ for $y\in \MN$.

\item $G\in C^0$ in $x$ and $y$ except on the diagonal of $M\times M$.

\item For any function $\phi$ where the following integrals make sense,
\begin{multline*}
\phi(x) = \int_M (-\Delta \phi +\alpha\phi)(y) G(x,y) dV(y) + \int_\MN (\partial_\nu \phi + \beta \phi)(y) G(x,y) dV(y) \\ - \int_\MD \phi(y) \partial_\nu G(x,y) dV(y)
\end{multline*} (We call this the definition of a Green's function for $L$.)

\item $G(x,y)>0$ for all $x,y$ such that $x,y \not\in \partial M$.

\item If $G(x,y) = 0$ (and so assume $y\in \partial M$), then $\partial_\nu G(x,y) <0$.

\item $\partial_\nu G(x,y) < 0$ for $y \in \MD$ and $G(x,y) > 0$ for $y\in \MN$.

\end{enumerate}
\end{thm}
\begin{proof}
For $x,y\in M$, let $r = d(x,y)$. We define
\[
H(x,y) = [(n-2)\omega_{n-1}]^{-1} r^{2-n} f(r)
\] where $\omega_{n-1}$ is the volume of a $n-1$ ball and $f(r)$ is some smooth positive decreasing function which is 1 in a neighborhood of 0 and 0 for $r> \textrm{inj}(x) (k+1)^{-1}$ where $\N \ni k >n/2$. The injectivity radius is positive at each point $x$ since $M$ is a compact manifold. The function $H$ is as smooth as the metric away from $r=0$, and so $\Delta H$ exists in a weak sense on $M\setminus B_x(\epsilon)$.

Green's formula is a standard result. It says that for functions $\phi$ that are regular enough,
\[
\phi(x) = \int_M H(x,y) \Delta \phi(y) dV(y) - \int_M \Delta_y H(x,y) \phi(y)dV(y)
\] where $\Delta_y$ means the standard Laplacian in the $y$ variable. It is proven by computing $\int_{M\setminus B_x(\epsilon)} H(x,y) \Delta \phi(y) dV(y)$, integrating by parts twice and then letting $\epsilon\to 0$. Since our ball $B_x(\epsilon)$ is essentially a coordinate ball in normal coordinates, and because metrics go to the Euclidean metric as $\epsilon\to 0$, the boundary terms converge as in the regular proof for this result. For $\phi$, ``regular enough'' means that the  boundary integrals from the proof make sense and have the appropriate limits. So, for instance, $\phi\in W^{2,1}  \cap C^0$ would be sufficient. In particular, $\phi(y) = H(y,z)$ would also work, for $z\neq x$.

Let $\Delta^*$ be the formal $L^2$ adjoint of $\Delta$ on $M$, i.e., $\langle \Delta^* f, g\rangle = \langle f, \Delta g \rangle$ for appropriate functions $f,g$. This is a well defined functional by the Riesz Representation Theorem. Green's theorem could then be interpreted as saying
\[
\Delta^*_y H(x,y) = \Delta_y H(x,y) + \delta_x^y,
\] where $\delta_x^y$ is the Dirac delta function.

Using this, we can rewrite Green's formula as
\begin{align}\label{otherGreens}
\phi(x) &= \int_M \Delta^*_y H(x,y) \phi(y) dV(y) - \int_M \Delta_y H(x,y) \phi(y)dV(y) \\
&=\Delta_x^* \int_M H(x,y) \phi(y) dV(y) - \int_M \Delta_x H(x,y) \phi(y) dV(y) \label{otherGreens}
\end{align} by the symmetry of $H(x,y)$.

We define
\begin{align*}
\Gamma(x,y) = \Gamma_1(x,y) &= (-\Delta^*_y+\alpha(y)) H(x,y)\\
\Gamma_{i+1}(x,y) &= \int_M \Gamma_i(x,z) \Gamma(z,y) dV(z).
\end{align*} For $\N \ni k>n/2$, let
\begin{equation}\label{GDefinition}
G(x,y) = H(x,y) + \sum_{i=1}^k \int_M (-1)^i \Gamma_i(x,z)H(z,y) dV(z) + F(x,y)
\end{equation} where $F$ satisfies
\begin{equation}\label{eq:FDefinition}
\begin{array}{rll} -\Delta_y F(x,y) + \alpha(y) F(x,y) &= (-1)^{k+1}\Gamma_{k+1}(x,y) & \textrm{   on } M \\
\partial_\nu F(x,y) + \beta(y) F(x,y) &=0 & \textrm{   on } \MN \\
 F(x,y) &= 0 & \textrm{   on } \MD.  \end{array}
\end{equation} The choice of $f(r)$ we made earlier guarantees that the non-$F(x,y)$ terms of $G(x,y)$ are identically zero in a neighborhood of the boundary, and so this $G(x,y)$ fulfills (a).

The $\Gamma_i$ were defined in this way so that by \cite[Prop 4.12]{Aubin98}, $\Gamma_{k+1} \in C^0\subset L^p$ in both $x$ and $y$. Thus a solution $F\in W^{2,p}$ of \eqref{eq:FDefinition} by \cite[Lem B.6.]{HT13}, where the regularity is only for the $y$ variable.

Since $H(x,y)$ is $W^{2,p}$ in both variables away from the diagonal, the second term in $G(x,y)$ is $W^{2,p}$ away from the diagonal. In the $y$ variable this is because $H(z,y)$ is $W^{2,p}$ in $y$. In the $x$ variable this is because $\nabla_x$ is only taking a derivative of $H(x,\cdot)$. We then just need to show that $F(x,y)$ is continuous in $x$ to show (b). To do this, we apply the standard elliptic estimate from \cite[Lem B.8.]{HT13}
\begin{multline*}
\|F(x,y) -F(z,y) \|_{\infty} \leq C \|F(x,y) - F(z,y)\|_{{2,p}} \\ \leq C\|\Gamma_{k+1}(x,y)-\Gamma_{k+1}(z,y)\|_{p} \leq C\|\Gamma_{k+1}(x,y)-\Gamma_{k+1}(z,y)\|_{\infty}
\end{multline*} where the boundary terms disappear by our choice of $F$. Thus, because $\Gamma_{k+1}(x,y)$ is continuous in $x$, so is $F(x,y)$. This completes (b).

We apply the operator $(-\Delta_y + \alpha(y))$ to both sides of Equation \eqref{GDefinition} and use identity \eqref{otherGreens}. Suppressing the variables, we get
\begin{align*}
(-\Delta+\alpha)G &= (-\Delta+\alpha)H + (-\Delta+\alpha)\left(\sum_{i=1}^k \int_M (-1)^i \Gamma_i H\right) + (-\Delta+\alpha)F \\
&= \delta + (-\Delta^* + \alpha)H + \sum_{i=1}^k (-1)^i \Gamma_i \\
&\hspace{20mm} + \sum_{i=1}^k \int_M (-1)^i(-\Delta^* + \alpha)(H) \Gamma_i + (-1)^{k+1}\Gamma_{k+1} \\
&= \delta + \Gamma_1 + \sum_{i=1}^k (-1)^i \Gamma_i + \sum_{i=1}^k (-1)^i \Gamma_{i+1} + (-1)^{k+1}\Gamma_{k+1} \\
&= \delta.
\end{align*}

This gives us that $(-\Delta_y+\alpha(y)) G(x,y) = \delta_x^y$. We then calculate for any $\phi\in C^2$, again suppressing variables,
\begin{align*}
\phi &= \int_M \phi (-\Delta +\alpha)G \\
&= \int_M(-\Delta+\alpha)\phi G + \int_{\partial M} \partial_\nu \phi G + \phi \partial_\nu G \\
&= \int_M (-\Delta +\alpha)\phi G + \int_\MN (\partial_\nu + \beta)\phi G - \int_\MD \phi \partial_\nu G
\end{align*} which is part (c). The second line came from integrating by parts twice. The third line is by substituting in the boundary conditions for $G$. While this calculation is only valid for $C^2$ functions, by a standard density argument (e.g., \cite[Prop 4.14]{Aubin98}), we can say that the equality holds for any functions $\phi$ where the integrals make sense.

Clearly $G(x,y) \geq 0$ everywhere. Indeed, for a fixed $x$, $G(x,y)$ satisfies
\[
(-\Delta_y+\alpha(y)) G(x,y) = 0
\]on $M\setminus B_x(\epsilon)$. By the maximum principle in \cite{HT13}, $G(x,y)\geq 0$. We also get that $G(x,y)$ is $W^{2,p}$ in $y$, away from $x=y$.

In fact, it is $0$ only on the boundary. Suppose it was $0$ elsewhere. Then by \cite[Thm 8.19]{GT}, the strong maximum principle, since $G(x,y)$ is $W^{2,p} \subset W^{1,2}$, away from $x=y$, $G$ must be constant away from the diagonal. However, it cannot be identically zero because for $y$ near $x$, $G(x,y)$ goes to infinity by \cite[Prop 4.12]{Aubin98}. (In particular, that proposition implies that $H(x,y)$ remains the leading term of $G(x,y)$.) Thus we have (d).

Assume $G(x,y_0) = 0$ for $y_0 \in \partial M$. The Hopf lemma, as usually stated, requires that $g\in C^2$. However, the proof in \cite{GT} can be easily generalized to the case when $g\in W^{2,p}$, $p>n$. Thus, since $LG=0$ and $G(x,y)>0$ for $y \in M$ near $y_0$, we have (e).

Part (f) immediately follows from parts (a) and (e).

\end{proof}

\bibliographystyle{alpha}
\bibliography{KnownResults}
\end{document}